\documentclass{article}

\usepackage{PRIMEarxiv}

\usepackage[utf8]{inputenc} 
\usepackage[T1]{fontenc}    
\usepackage{hyperref}       
\usepackage{url}            
\usepackage{booktabs}       
\usepackage{amsmath}
\usepackage{amsthm}
\usepackage{amssymb}
\usepackage{amsfonts}
\usepackage{mathtools}
\usepackage{amsfonts}       
\usepackage{nicefrac}       
\usepackage{microtype}      
\usepackage{fancyhdr}       
\usepackage{graphicx}       
\usepackage{times}
\usepackage{subfig}
\usepackage{xcolor}
\usepackage{array}
\usepackage{cite}
\usepackage{multirow}
\usepackage{adjustbox}
\usepackage{algorithm}
\usepackage{algpseudocode}
\usepackage[font=footnotesize,labelfont=bf]{caption}
\usepackage{tikz-network}
\usepackage{enumitem}

\newtheorem{theorem}{Theorem}

\newtheorem{proposition}[theorem]{Proposition}

\newtheorem{remark}{Remark}

\title{
From Winding-Fault Geometry to Reliability: Estimation and Prognosis of Stator Inter-Turn Faults
}

\author{
  Bambang L. Widjiantoro, Katherin Indriawati, Moh Kamalul Wafi \\
  Department of Engineering Physics \\
  Institut Teknologi Sepuluh Nopember \\
  Surabaya, Jawa Timur, Indonesia\\
  \texttt{\{bambang.lw, katherin, kamalul.wafi\}@its.ac.id} \\
   \And
  Syahrul Munir \\
  Physics Department \\
  Universitas Pembangunan Nasional "Veteran" Jawa Timur \\
  Surabaya, Jawa Timur, Indonesia \\
  \texttt{syahrul.munir.fisika@upnjatim.ac.id} \\
}

\begin{document}
\maketitle

\begin{abstract}
This paper develops an integrated framework for estimation, prognosis, and
reliability assessment of stator inter-turn faults in induction motors.
The fault is characterized by its severity, defined as the fraction of
short-circuited turns, and its spatial orientation. A geometric fault model
shows that the resulting output signature is affine in the fault severity
and exhibits a fundamental spatial periodicity. Exploiting this structure,
an augmented-state particle filter jointly estimates the nonlinear
electromechanical state and the unknown fault severity while quantifying
posterior uncertainty. The estimated degradation is then propagated using
four prognostic models: linear trend, Holt exponential smoothing, Bayesian
degradation, and particle-based forecasting. Their predictions are connected
to threshold-crossing remaining useful life (RUL), first-passage reliability,
degradation-dependent hazard reliability, and a Weibull lifetime benchmark,
thereby providing both deterministic and probabilistic health assessments.
Numerical results demonstrate accurate online fault estimation and output
reconstruction, characterize the effects of degradation pattern and
prediction horizon on prognosis, and show consistent reliability and
threshold-crossing predictions. Moreover, estimation accuracy remains
comparable across the distinct fault orientations induced by the spatial
periodicity. The resulting framework provides a unified connection from
physics-based inter-turn fault modeling to online diagnosis, degradation
prognosis, and reliability assessment.
\end{abstract}
\allowdisplaybreaks


\section{Introduction}
\label{sec:introduction}

Induction motors remain central to industrial electromechanical systems, and
their reliable operation has motivated extensive research on condition
monitoring and fault diagnosis \cite{R1,WAFI-JRC,R3,R4}. Among electrical failure
modes, stator-winding faults are particularly important because insulation
degradation may initiate an inter-turn short circuit (ITSC), whose continued
development can produce local heating and progressively more severe winding
damage \cite{R5,R6,R7}. Consequently, early detection and severity
quantification of ITSC faults have received considerable attention
\cite{R8,R9,Wafi-CoDIT,R11}. Detecting the presence of a fault, however, addresses
only the current condition of the machine. Condition-based and predictive
maintenance additionally require prediction of future degradation, remaining
useful life (RUL), and the associated uncertainty and failure risk
\cite{R12,R13,R14,R15}.

A central difficulty is that a stator ITSC has both a \emph{magnitude} and a
\emph{spatial structure}. Existing fault models commonly characterize the
fault through quantities related to the fraction of short-circuited turns
together with the location of the faulty winding \cite{R8,R9,R16,R17}.
Accordingly, the number of affected turns determines the severity of the
winding degradation, while the fault location influences how the resulting
electrical asymmetry appears in measured currents and voltages. This
dependence has motivated analytical machine models, parameter-identification
methods, and signal-based techniques for detecting, locating, and quantifying
stator ITSC faults \cite{R6,R7,Wafi-CoDIT,R11,R16}. A useful health variable for
prognosis should therefore remain physically interpretable while retaining
the connection between fault severity, spatial orientation, and the measured
motor response.

Model-based fault diagnosis provides a natural mechanism for exploiting this
physical structure. Observer- and parameter-estimation approaches have been
used to reconstruct winding faults from measured motor variables
\cite{R9,R18,R19}, while Bayesian state estimation provides a systematic
framework for representing uncertainty in unknown states and parameters.
Particle filtering, in particular, approximates the recursive Bayesian
posterior by a weighted particle population and is applicable to nonlinear
and non-Gaussian state-space models \cite{Wafi-Elham,R21}. This property has also
made particle filtering attractive for prognostics, where the posterior
health state can be propagated forward to obtain distributions of future
degradation and RUL \cite{R22,R23,R24,R25}.

Prognosis nevertheless introduces a modeling problem distinct from online
fault estimation. Future degradation must be inferred without future
measurements, and the resulting prediction can depend strongly on the assumed
degradation dynamics and prediction horizon \cite{R14,R15,R26}. Local trend
extrapolation and exponential smoothing provide computationally simple
deterministic forecasts \cite{R27,Wafi-AIP,Wafi-AIP-arXiv}, whereas stochastic degradation
models permit uncertainty in the degradation rate and future evolution to be
carried explicitly into RUL prediction \cite{R26,R29,R30}. Such uncertainty
is particularly important for long-horizon prognosis, for which a single
point estimate may conceal substantial variability in the predicted
threshold-crossing time \cite{R15,R22}.

Reliability provides the subsequent connection between predicted degradation
and failure risk. When failure is associated with a critical degradation
level, RUL can be formulated through the first-passage or first-hitting time
of the degradation process \cite{R29,R30}. Alternatively, condition
information can enter reliability through a degradation-dependent hazard
model, allowing the predicted health state to modify the instantaneous
failure rate \cite{R31,R32}. Classical lifetime models such as the Weibull
distribution provide a complementary population-level description through
their survival and hazard functions \cite{Wafi-Quadruple}. These formulations answer
different questions: threshold crossing concerns when a prescribed
degradation boundary is reached, first-passage reliability concerns survival
of the complete future degradation path, and hazard-based reliability
accumulates condition-dependent failure risk. More generally, uncertainty
propagation from diagnosis through prognosis to reliability remains an
important issue in practical PHM \cite{R15,R34}.

Motivated by these considerations, this paper develops an integrated
estimation--prognosis--reliability framework around a common physical health
variable. Let $\eta_{cc}\in[0,1]$ denote the fraction of short-circuited
stator turns and let $\gamma_{cc}$ denote the spatial orientation of the
equivalent short-circuit winding. The fault-induced current is represented
through a geometric orientation model and factorized into the product of
$\eta_{cc}$ and an orientation- and operating-condition-dependent signature.
This yields a measured-output model that is affine in the unknown fault
severity. The same variable $\eta_{cc}$ is then carried throughout the
remainder of the framework: it is estimated online, propagated into the
future, and ultimately mapped into RUL and reliability quantities.

The main contributions of this paper are summarized as follows:
\begin{itemize}
    \item A geometric representation of stator inter-turn faults is developed
    in terms of the physically interpretable pair
    $(\eta_{cc},\gamma_{cc})$. The resulting fault contribution is affine in
    $\eta_{cc}$, and its spatial periodicity is characterized, identifying the
    distinct fault orientations represented by the first-harmonic model.

    \item An augmented-state particle-filter formulation is developed for
    joint estimation of the nonlinear electromechanical state and the
    constrained fault severity $\eta_{cc}\in[0,1]$. The posterior particle
    population provides both an online severity estimate and its uncertainty,
    rather than only a point-valued fault indicator.

    \item Four prognosis mechanisms are formulated on the same PF-based
    health information: linear trend, Holt exponential smoothing, Bayesian
    degradation, and particle-based forecasting. This provides deterministic
    and probabilistic descriptions of future winding degradation while
    separating the estimation stage from assumptions concerning long-term
    fault evolution.

    \item The resulting forecasts are connected to complementary reliability
    measures, including threshold-crossing RUL, first-passage reliability,
    degradation-dependent hazard reliability, and a Weibull lifetime
    benchmark. For probabilistic forecasts, both pathwise threshold survival
    and uncertainty-averaged degradation hazards are considered.

    \item Numerical studies evaluate the complete framework under gradual,
    accelerating, and abrupt degradation, different prediction horizons, and
    distinct fault orientations. The results demonstrate accurate online
    severity estimation and output reconstruction, quantify the effect of
    degradation dynamics on prognosis, and illustrate how different
    reliability formulations interpret the same predicted winding condition.
\end{itemize}

The remainder of the paper is organized as follows.
Section~\ref{sec:electromechanical_model} presents the induction-motor model.
Section~\ref{sec:fault_model} develops the stator inter-turn fault model and
its spatial representation. Section~\ref{sec:particle_filtering} introduces
the augmented-state particle filter. Section~\ref{sec:fault_prognosis}
develops the four fault-prognosis methods, and
Section~\ref{sec:reliability} formulates the corresponding reliability
measures. Section~\ref{sec:numerical_results} presents the numerical results,
and the paper is concluded in the final section.

\section{Electromechanical Model of the Induction Motor}
\label{sec:electromechanical_model}

This section derives the nominal electromechanical model of the induction
motor in the synchronous $(d,q)$ reference frame. The resulting model
describes the common-mode dynamics of the machine and provides the baseline
representation upon which the stator winding-fault model is constructed in
the following section. The derivation begins with the stator (subscript $s$)
and rotor (subscript $r$) voltage equations, followed by the corresponding
flux-linkage relations. The rotor currents are then eliminated to obtain
differential equations in terms of the stator currents and rotor fluxes.
Finally, the electrical dynamics are coupled with the mechanical equation
of the rotor to obtain the complete nonlinear state-space model.

The stator voltage equations in the synchronous $(d,q)$ reference frame are
\begin{subequations}
\begin{align}
    V_{sd}
    &=
    R_s i_{sd}
    +
    \frac{d}{dt}\lambda_{sd}
    -
    \omega_s \lambda_{sq},
    \label{eq:vsd_basic}
    \\
    V_{sq}
    &=
    R_s i_{sq}
    +
    \frac{d}{dt}\lambda_{sq}
    +
    \omega_s \lambda_{sd},
    \label{eq:vsq_basic}
\end{align}
\end{subequations}
where $V_{sd}$ and $V_{sq}$ denote the stator voltages along the $d$- and
$q$-axes, respectively; $i_{sd}$ and $i_{sq}$ are the corresponding stator
currents; $\lambda_{sd}$ and $\lambda_{sq}$ are the stator flux linkages; and
$R_s>0$ is the stator resistance. The quantity $\omega_s$ denotes the
synchronous electrical angular velocity of the rotating reference frame.

Similarly, the rotor voltage equations in the synchronous $(d,q)$ reference
frame are
\begin{subequations}
\begin{align}
    V_{rd}
    &=
    R_r i_{rd}
    +
    \frac{d}{dt}\lambda_{rd}
    -
    (\omega_s-p\omega_r)\lambda_{rq},
    \label{eq:vrd_basic}
    \\
    V_{rq}
    &=
    R_r i_{rq}
    +
    \frac{d}{dt}\lambda_{rq}
    +
    (\omega_s-p\omega_r)\lambda_{rd},
    \label{eq:vrq_basic}
\end{align}
\end{subequations}
where $V_{rd}$ and $V_{rq}$ denote the rotor voltages along the $d$- and
$q$-axes, respectively; $i_{rd}$ and $i_{rq}$ are the corresponding rotor
currents; $\lambda_{rd}$ and $\lambda_{rq}$ are the rotor flux linkages;
$R_r>0$ is the rotor resistance; $p$ is the number of pole pairs; and
$\omega_r$ is the mechanical angular velocity of the rotor. Consequently,
$p\omega_r$ represents the rotor electrical angular velocity, and
$\omega_s-p\omega_r$ is the relative electrical angular velocity between
the synchronous reference frame and the rotor.

For a squirrel-cage induction motor, the rotor windings are short-circuited
through the end rings and no external voltage is applied to the rotor.
Therefore, $V_{rd}=V_{rq}=0.$
It follows from \eqref{eq:vrd_basic}--\eqref{eq:vrq_basic} that the rotor
flux dynamics satisfy
\begin{subequations}
\begin{align}
    \frac{d}{dt}\lambda_{rd}
    &=
    -R_r i_{rd}
    +
    (\omega_s-p\omega_r)\lambda_{rq},
    \label{eq:lrd_basic}
    \\
    \frac{d}{dt}\lambda_{rq}
    &=
    -R_r i_{rq}
    -
    (\omega_s-p\omega_r)\lambda_{rd}.
    \label{eq:lrq_basic}
\end{align}
\end{subequations}
The rotor flux dynamics in \eqref{eq:lrd_basic}--\eqref{eq:lrq_basic}
still depend on the rotor currents $i_{rd}$ and $i_{rq}$. To express these
dynamics in terms of the stator currents and rotor flux linkages, we next
use the magnetic coupling relations between the stator and rotor circuits.
The stator and rotor flux linkages satisfy
\begin{subequations}
\begin{align}
    \lambda_{sd}
    &=
    L_s i_{sd}
    +
    L_m i_{rd},
    \label{eq:lsd}
    \\
    \lambda_{sq}
    &=
    L_s i_{sq}
    +
    L_m i_{rq},
    \label{eq:lsq}
    \\
    \lambda_{rd}
    &=
    L_r i_{rd}
    +
    L_m i_{sd},
    \label{eq:lrd}
    \\
    \lambda_{rq}
    &=
    L_r i_{rq}
    +
    L_m i_{sq},
    \label{eq:lrq}
\end{align}
\end{subequations}
where $L_s>0$ and $L_r>0$ denote the stator and rotor self-inductances,
respectively, while $L_m>0$ denotes the mutual inductance between the
stator and rotor windings. In particular, the terms involving $L_m$
represent the magnetic coupling between the stator and rotor circuits.
Since the rotor currents are not selected as state variables, we eliminate
$i_{rd}$ and $i_{rq}$ from the subsequent dynamics. Solving
\eqref{eq:lrd} and \eqref{eq:lrq} for the rotor currents gives
\begin{subequations}
\begin{align}
    i_{rd}
    &=
    \frac{1}{L_r}
    \left(
    \lambda_{rd}-L_m i_{sd}
    \right),
    \label{eq:ird}
    \\
    i_{rq}
    &=
    \frac{1}{L_r}
    \left(
    \lambda_{rq}-L_m i_{sq}
    \right).
    \label{eq:irq}
\end{align}
\end{subequations}
Substituting \eqref{eq:ird} and \eqref{eq:irq} into
\eqref{eq:lrd_basic} and \eqref{eq:lrq_basic}, respectively, eliminates
the rotor currents from the rotor flux dynamics and gives
\begin{subequations}
\begin{align}
    \frac{d}{dt}\lambda_{rd}
    &=
    \frac{R_r}{L_r}
    \left(
    L_m i_{sd}-\lambda_{rd}
    \right)
    +
    (\omega_s-p\omega_r)\lambda_{rq},
    \label{eq:lrd_final}
    \\
    \frac{d}{dt}\lambda_{rq}
    &=
    \frac{R_r}{L_r}
    \left(
    L_m i_{sq}-\lambda_{rq}
    \right)
    -
    (\omega_s-p\omega_r)\lambda_{rd}.
    \label{eq:lrq_final}
\end{align}
\end{subequations}
Thus, the rotor flux dynamics are expressed entirely in terms of the stator
currents $i_{sd}$ and $i_{sq}$, the rotor flux linkages $\lambda_{rd}$ and
$\lambda_{rq}$, and the rotor angular velocity $\omega_r$, without requiring
the rotor currents $i_{rd}$ and $i_{rq}$ as additional dynamic variables.

We next derive the dynamics of the stator currents. Starting with the
$d$-axis, substituting the stator flux relation \eqref{eq:lsd} into
\eqref{eq:vsd_basic} gives
\begin{align}
    V_{sd}
    &=
    R_s i_{sd}
    +
    \frac{d}{dt}
    \left(
    L_s i_{sd}+L_m i_{rd}
    \right)
    -
    \omega_s\lambda_{sq}
    \nonumber\\
    &=
    R_s i_{sd}
    +
    L_s\frac{d}{dt}i_{sd}
    +
    L_m\frac{d}{dt}i_{rd}
    -
    \omega_s\lambda_{sq}.
    \label{eq:vsd_step1}
\end{align}
The resulting expression contains the derivative of the rotor current
$i_{rd}$. To eliminate this quantity, differentiating \eqref{eq:ird}
with respect to time yields
\begin{equation}
    \frac{d}{dt}i_{rd}
    =
    \frac{1}{L_r}
    \left(
    \frac{d}{dt}\lambda_{rd}
    -
    L_m\frac{d}{dt}i_{sd}
    \right).
    \label{eq:dird}
\end{equation}
Substituting \eqref{eq:dird} into \eqref{eq:vsd_step1} eliminates
$\dot{i}_{rd}$ and gives
\begin{align}
    V_{sd}
    &=
    R_s i_{sd}
    +
    L_s\frac{d}{dt}i_{sd}
    +
    \frac{L_m}{L_r}\frac{d}{dt}\lambda_{rd}
    -
    \frac{L_m^2}{L_r}\frac{d}{dt}i_{sd}
    -
    \omega_s\lambda_{sq}
    \nonumber\\
    &=
    R_s i_{sd}
    +
    \left(
    L_s-\frac{L_m^2}{L_r}
    \right)
    \frac{d}{dt}i_{sd}
    +
    \frac{L_m}{L_r}\frac{d}{dt}\lambda_{rd}
    -
    \omega_s\lambda_{sq}.
    \label{eq:vsd_step2}
\end{align}
The expression in \eqref{eq:vsd_step2} still contains
$\dot{\lambda}_{rd}$. Using the rotor flux dynamics derived in
\eqref{eq:lrd_final}, we obtain
\begin{align}
    V_{sd}
    &=
    R_s i_{sd}
    +
    \left(
    L_s-\frac{L_m^2}{L_r}
    \right)
    \frac{d}{dt}i_{sd}
    +
    \frac{L_m}{L_r}
    \left[
    \frac{R_r}{L_r}
    \left(
    L_m i_{sd}-\lambda_{rd}
    \right)
    +
    (\omega_s-p\omega_r)\lambda_{rq}
    \right]
    -
    \omega_s\lambda_{sq}
    \nonumber\\
    &=
    R_s i_{sd}
    +
    \left(
    L_s-\frac{L_m^2}{L_r}
    \right)
    \frac{d}{dt}i_{sd}
    +
    \frac{L_m^2R_r}{L_r^2}i_{sd}
    -
    \frac{L_mR_r}{L_r^2}\lambda_{rd}
    +
    \frac{L_m}{L_r}\omega_s\lambda_{rq}
    -
    \frac{L_m}{L_r}p\omega_r\lambda_{rq}
    -
    \omega_s\lambda_{sq}.
    \label{eq:vsd_step3}
\end{align}
Since
\begin{equation*}
    \lambda_{sq}
    =
    L_s i_{sq}
    +
    L_m i_{rq}
    =
    L_s i_{sq}
    +
    \frac{L_m}{L_r}
    \left(
    \lambda_{rq}-L_m i_{sq}
    \right),
\end{equation*}
we have
\begin{equation}
    \lambda_{sq}
    =
    \left(
    L_s-\frac{L_m^2}{L_r}
    \right)i_{sq}
    +
    \frac{L_m}{L_r}\lambda_{rq}.
    \label{eq:lsq_reduced}
\end{equation}
Substituting \eqref{eq:lsq_reduced} into \eqref{eq:vsd_step3}
eliminates $\lambda_{sq}$. Moreover, the two terms containing
$\omega_s\lambda_{rq}$ cancel, resulting in
\begin{align}
    V_{sd}
    &=
    R_s i_{sd}
    +
    \left(
    L_s-\frac{L_m^2}{L_r}
    \right)
    \frac{d}{dt}i_{sd}
    +
    \frac{L_m^2R_r}{L_r^2}i_{sd}
    -
    \frac{L_mR_r}{L_r^2}\lambda_{rd}
    \nonumber\\
    &\quad
    -
    \frac{L_m}{L_r}p\omega_r\lambda_{rq}
    -
    \omega_s
    \left(
    L_s-\frac{L_m^2}{L_r}
    \right)i_{sq}.
    \label{eq:vsd_step4}
\end{align}
To simplify the resulting expressions, define the leakage coefficient
\begin{equation}
    \sigma
    =
    \frac{L_sL_r-L_m^2}{L_sL_r},
    \label{eq:sigma}
    \end{equation}
    which satisfies
    \begin{equation}
    L_s-\frac{L_m^2}{L_r}
    =
    \sigma L_s.
\end{equation}
Using this identity in \eqref{eq:vsd_step4} and solving for
$\dot{i}_{sd}$ yields the $d$-axis stator current dynamics
\begin{align}
    \frac{d}{dt}i_{sd}
    &=
    \frac{V_{sd}}{\sigma L_s}
    -
    \frac{L_r^2R_s+L_m^2R_r}{\sigma L_sL_r^2}i_{sd}
    +
    \omega_s i_{sq}
    +
    \frac{L_mR_r}{\sigma L_sL_r^2}\lambda_{rd}
    +
    \frac{pL_m\omega_r}{\sigma L_sL_r}\lambda_{rq}.
    \label{eq:isd_final}
\end{align}

We proceed similarly for the $q$-axis stator current. Substituting the
stator flux relation \eqref{eq:lsq} into the $q$-axis stator voltage
equation \eqref{eq:vsq_basic} gives
\begin{align}
    V_{sq}
    &=
    R_s i_{sq}
    +
    \frac{d}{dt}
    \left(
    L_s i_{sq}+L_m i_{rq}
    \right)
    +
    \omega_s\lambda_{sd}
    \nonumber\\
    &=
    R_s i_{sq}
    +
    L_s\frac{d}{dt}i_{sq}
    +
    L_m\frac{d}{dt}i_{rq}
    +
    \omega_s\lambda_{sd}.
    \label{eq:vsq_step1}
\end{align}
The resulting equation contains the derivative of the rotor current
$i_{rq}$. Differentiating \eqref{eq:irq} with respect to time gives
\begin{equation}
    \frac{d}{dt}i_{rq}
    =
    \frac{1}{L_r}
    \left(
    \frac{d}{dt}\lambda_{rq}
    -
    L_m\frac{d}{dt}i_{sq}
    \right).
    \label{eq:dirq}
\end{equation}
Substituting \eqref{eq:dirq} into \eqref{eq:vsq_step1} eliminates
$\dot{i}_{rq}$ and yields
\begin{align}
    V_{sq}
    &=
    R_s i_{sq}
    +
    L_s\frac{d}{dt}i_{sq}
    +
    \frac{L_m}{L_r}\frac{d}{dt}\lambda_{rq}
    -
    \frac{L_m^2}{L_r}\frac{d}{dt}i_{sq}
    +
    \omega_s\lambda_{sd}
    \nonumber\\
    &=
    R_s i_{sq}
    +
    \left(
    L_s-\frac{L_m^2}{L_r}
    \right)
    \frac{d}{dt}i_{sq}
    +
    \frac{L_m}{L_r}\frac{d}{dt}\lambda_{rq}
    +
    \omega_s\lambda_{sd}.
    \label{eq:vsq_step2}
\end{align}
The remaining rotor flux derivative $\dot{\lambda}_{rq}$ can be
eliminated using the rotor flux dynamics in \eqref{eq:lrq_final}.
Substitution into \eqref{eq:vsq_step2} gives
\begin{align}
    V_{sq}
    &=
    R_s i_{sq}
    +
    \left(
    L_s-\frac{L_m^2}{L_r}
    \right)
    \frac{d}{dt}i_{sq}
    +
    \frac{L_m}{L_r}
    \left[
    \frac{R_r}{L_r}
    \left(
    L_m i_{sq}-\lambda_{rq}
    \right)
    -
    (\omega_s-p\omega_r)\lambda_{rd}
    \right]
    +
    \omega_s\lambda_{sd}
    \nonumber\\
    &=
    R_s i_{sq}
    +
    \left(
    L_s-\frac{L_m^2}{L_r}
    \right)
    \frac{d}{dt}i_{sq}
    +
    \frac{L_m^2R_r}{L_r^2}i_{sq}
    -
    \frac{L_mR_r}{L_r^2}\lambda_{rq}
    -
    \frac{L_m}{L_r}\omega_s\lambda_{rd}
    +
    \frac{L_m}{L_r}p\omega_r\lambda_{rd}
    +
    \omega_s\lambda_{sd}.
    \label{eq:vsq_step3}
\end{align}
Since
\begin{equation*}
    \lambda_{sd}
    =
    L_s i_{sd}
    +
    L_m i_{rd}
    =
    L_s i_{sd}
    +
    \frac{L_m}{L_r}
    \left(
    \lambda_{rd}-L_m i_{sd}
    \right),
\end{equation*}
we have
\begin{equation}
    \lambda_{sd}
    =
    \left(
    L_s-\frac{L_m^2}{L_r}
    \right)i_{sd}
    +
    \frac{L_m}{L_r}\lambda_{rd}.
    \label{eq:lsd_reduced}
\end{equation}
Substituting \eqref{eq:lsd_reduced} into \eqref{eq:vsq_step3}
eliminates $\lambda_{sd}$. The terms containing
$\omega_s\lambda_{rd}$ cancel, and hence
\begin{align}
    V_{sq}
    &=
    R_s i_{sq}
    +
    \left(
    L_s-\frac{L_m^2}{L_r}
    \right)
    \frac{d}{dt}i_{sq}
    +
    \frac{L_m^2R_r}{L_r^2}i_{sq}
    -
    \frac{L_mR_r}{L_r^2}\lambda_{rq}
    \nonumber\\
    &\quad
    +
    \frac{L_m}{L_r}p\omega_r\lambda_{rd}
    +
    \omega_s
    \left(
    L_s-\frac{L_m^2}{L_r}
    \right)i_{sd}.
    \label{eq:vsq_step4}
\end{align}
Using $L_s-L_m^2/L_r=\sigma L_s$
and solving \eqref{eq:vsq_step4} for $\dot{i}_{sq}$ yields the
$q$-axis stator current dynamics
\begin{align}
    \frac{d}{dt}i_{sq}
    &=
    \frac{V_{sq}}{\sigma L_s}
    -
    \frac{L_r^2R_s+L_m^2R_r}{\sigma L_sL_r^2}i_{sq}
    -
    \omega_s i_{sd}
    +
    \frac{L_mR_r}{\sigma L_sL_r^2}\lambda_{rq}
    -
    \frac{pL_m\omega_r}{\sigma L_sL_r}\lambda_{rd}.
    \label{eq:isq_final}
\end{align}

Having derived the stator current and rotor flux dynamics, it remains to
describe the mechanical motion of the rotor. The electrical and mechanical
dynamics are coupled through the electromagnetic torque, which is given by
\begin{equation}
    T_e
    =
    p\frac{L_m}{L_r}
    \left(
    \lambda_{rd}i_{sq}
    -
    \lambda_{rq}i_{sd}
    \right),
    \label{eq:torque}
\end{equation}
where $T_e$ denotes the electromagnetic torque generated by the motor.
The rotor angular velocity evolves according to the mechanical torque
balance
\begin{equation}
    J\frac{d}{dt}\omega_r
    =
    T_e-T_l,
    \label{eq:mechanical_basic}
\end{equation}
where $J>0$ denotes the total moment of inertia of the rotor and load, and
$T_l$ is the external load torque acting against the electromagnetic torque.
Equivalently,
\begin{equation}
    \frac{d}{dt}\omega_r
    =
    \frac{T_e-T_l}{J}.
\end{equation}
Substituting \eqref{eq:torque} into the mechanical equation eliminates
$T_e$ and yields the rotor-speed dynamics
\begin{equation}
    \frac{d}{dt}\omega_r
    =
    \frac{pL_m}{JL_r}
    \left(
    \lambda_{rd}i_{sq}
    -
    \lambda_{rq}i_{sd}
    \right)
    -
    \frac{T_l}{J}.
    \label{eq:wr_final}
\end{equation}

We have now obtained differential equations for the two stator currents,
the two rotor flux linkages, and the rotor angular velocity. Collecting
\eqref{eq:isd_final}, \eqref{eq:isq_final}, \eqref{eq:lrd_final},
\eqref{eq:lrq_final}, and \eqref{eq:wr_final}, the complete nonlinear
induction motor dynamics are given by
\begin{subequations}
\begin{align}
    \dot{i}_{sd}
    &=
    \frac{V_{sd}}{\sigma L_s}
    -
    \frac{L_r^2R_s+L_m^2R_r}{\sigma L_sL_r^2}i_{sd}
    +
    \omega_s i_{sq}
    +
    \frac{L_mR_r}{\sigma L_sL_r^2}\lambda_{rd}
    +
    \frac{pL_m\omega_r}{\sigma L_sL_r}\lambda_{rq},
    \\
    \dot{i}_{sq}
    &=
    \frac{V_{sq}}{\sigma L_s}
    -
    \frac{L_r^2R_s+L_m^2R_r}{\sigma L_sL_r^2}i_{sq}
    -
    \omega_s i_{sd}
    +
    \frac{L_mR_r}{\sigma L_sL_r^2}\lambda_{rq}
    -
    \frac{pL_m\omega_r}{\sigma L_sL_r}\lambda_{rd},
    \\
    \dot{\lambda}_{rd}
    &=
    \frac{R_r}{L_r}
    \left(
    L_m i_{sd}-\lambda_{rd}
    \right)
    +
    (\omega_s-p\omega_r)\lambda_{rq},
    \\
    \dot{\lambda}_{rq}
    &=
    \frac{R_r}{L_r}
    \left(
    L_m i_{sq}-\lambda_{rq}
    \right)
    -
    (\omega_s-p\omega_r)\lambda_{rd},
    \\
    \dot{\omega}_r
    &=
    \frac{pL_m}{JL_r}
    \left(
    \lambda_{rd}i_{sq}
    -
    \lambda_{rq}i_{sd}
    \right)
    -
    \frac{T_l}{J}.
    \end{align}
    \label{eq:im_full_model}
\end{subequations}

Finally, the nonlinear dynamics in \eqref{eq:im_full_model} can be expressed in
compact state-space form. Define the state vector
$x=[i_{sd},\,i_{sq},\,\lambda_{rd},\,\lambda_{rq},\,\omega_r]^{\top}
\in\mathbb{R}^{5}$, which consists of the two stator currents, the two rotor
flux linkages, and the rotor angular velocity. The stator voltages constitute
the input vector
$u=[V_{sd},\,V_{sq}]^{\top}\in\mathbb{R}^{2}$, while the load torque
$T_l(t)\in\mathbb{R}$ is treated as an external mechanical input to the motor
dynamics.

For the estimation framework developed in the following sections, the
available measurements are assumed to be the two stator currents and the
rotor angular velocity. Accordingly, the measured output is
$y=[i_{sd},\,i_{sq},\,\omega_r]^{\top}\in\mathbb{R}^{3}$. The induction
motor model can therefore be written compactly as
\begin{subequations}
\begin{align}
    \dot{x}(t)
    &=
    f(x(t),u(t),T_l(t)),
    \label{eq:im_compact_dynamics}
    \\
    y(t)
    &=
    Cx(t),
    \label{eq:im_compact_output}
\end{align}
\end{subequations}
where $C\in\mathbb{R}^{3\times5}$ is the output matrix, and the nonlinear
vector field
$f:\mathbb{R}^{5}\times\mathbb{R}^{2}\times\mathbb{R}\to\mathbb{R}^{5}$
is defined by the right-hand side of \eqref{eq:im_full_model}.

The model in
\eqref{eq:im_compact_dynamics}--\eqref{eq:im_compact_output}
constitutes the nominal, or common-mode, representation of the induction
motor. It describes the underlying electromechanical dynamics of the machine
independently of the winding-fault variables. Rather than modifying these
nominal dynamics directly, the stator inter-turn fault is characterized in
the following section through a differential-mode representation
parameterized by the fault severity and spatial location. This separation
provides a physically interpretable basis for the subsequent fault-estimation
and prognosis framework.

\section{Stator Inter-Turn Fault Model}
\label{sec:fault_model}

A stator inter-turn fault originates from the deterioration of the insulation
between neighboring turns of a stator winding. Once electrical contact occurs,
a subset of the winding turns becomes short-circuited and forms a closed
conducting path. The magnetic field generated by the machine then induces a
circulating current through this path. Since the resistance of the
short-circuited turns is relatively small, this circulating current may become
significant even at an early stage of the fault. The resulting local heating
can further damage the insulation and thereby promote progressive winding
degradation.

The nominal model developed in the previous section describes the common-mode
electromechanical dynamics of the induction motor. To characterize the
additional electrical behavior produced by the winding short circuit, we
introduce a differential-mode fault representation. The key idea is to replace
the localized short-circuited turns by an equivalent symmetrized short-circuit
winding whose first-harmonic electromagnetic effect is characterized by two
physically meaningful parameters: the fault severity $\eta_{cc}$ and the
spatial orientation $\gamma_{cc}$ of the short-circuit winding.

\subsection{Fault Severity and Spatial Orientation}

Let $n_s$ denote the number of turns of a healthy stator phase and let
$n_{cc}$ denote the number of turns involved in the short circuit. The
dimensionless fault-severity parameter is defined as
\begin{equation}
    \eta_{cc}
    :=
    \frac{n_{cc}}{n_s},
    \qquad
    0\leq \eta_{cc}\leq 1.
    \label{eq:fault_severity}
\end{equation}
Thus, $\eta_{cc}=0$ corresponds to the healthy winding, whereas
$\eta_{cc}>0$ indicates the presence of an inter-turn short circuit.
Increasing values of $\eta_{cc}$ represent an increasing fraction of
short-circuited turns and therefore provide a natural degradation variable
for the subsequent prognosis and reliability analysis.

The spatial orientation of the equivalent short-circuit winding is described
by the mechanical angular parameter $\gamma_{cc}$. Accordingly,
$p\gamma_{cc}$ represents its corresponding electrical orientation, where
$p$ is the number of pole pairs. In contrast to a phase-level description,
$\gamma_{cc}$ identifies the symmetry axis of the equivalent short-circuit
winding and therefore characterizes its spatial orientation.
For a stator containing $n_e$ uniformly distributed slots, the admissible spatial orientations may be represented by the finite set
\begin{equation}
    \gamma_{cc}
    \in
    \left\{
    k\frac{2\pi}{n_e}:
    k=0,\ldots,n_e-1
    \right\}.
    \label{eq:fault_orientation_set}
\end{equation}
The parameter $\eta_{cc}$ therefore quantifies \emph{how severe} the winding
fault is, while $\gamma_{cc}$ specifies the spatial orientation of the
corresponding short-circuit axis.

Although \eqref{eq:fault_orientation_set} contains $n_e$ admissible slot
orientations, the first-harmonic fault representation possesses an additional
spatial periodicity. Consequently, distinct physical slot orientations may
generate identical fault signatures, as characterized in the following
subsection.

\subsection{Geometric Representation of the Fault Orientation}

Having characterized the fault severity $\eta_{cc}$ and spatial orientation
$\gamma_{cc}$, we next examine how $\gamma_{cc}$ determines the
fault-induced current $i_{cc,dq}$. In particular, the spatial orientation
$\gamma_{cc}$ must be represented through a directional projection
$\mathcal{Q}(\cdot)$ and subsequently expressed in the synchronous $(d,q)$
reference frame through the rotational transformation $\mathcal{P}(\cdot)$.

To account for the change of orientation in the $(d,q)$ plane, define the
rotational transformation
\begin{equation}
    \mathcal{P}(\theta)
    :=
    \begin{bmatrix}
        \cos\theta & -\sin\theta\\
        \sin\theta &  \cos\theta
    \end{bmatrix}
    \in\mathbb{R}^{2\times2}.
    \label{eq:P_matrix}
\end{equation}
The matrix $\mathcal{P}(\theta)$ rotates a $(d,q)$ vector through the angle
$\theta$, allowing the stator-fixed short-circuit winding orientation to be
expressed in the synchronous $(d,q)$ reference frame.

The spatial orientation of the short-circuit winding is represented through
the directional operator
\begin{equation}
    \mathcal{Q}(\theta)
    :=
    \begin{bmatrix}
        \cos^2\theta &
        \cos\theta\sin\theta\\
        \cos\theta\sin\theta &
        \sin^2\theta
    \end{bmatrix}
    \in\mathbb{R}^{2\times2}.
    \label{eq:Q_matrix}
\end{equation}
Unlike $\mathcal{P}(\theta)$, which rotates a $(d,q)$ vector through the
angle $\theta$, $\mathcal{Q}(\theta)$ selects the component aligned with
the short-circuit winding axis specified by $\theta$. To make this
directional role explicit, define
$q(\theta):=[\cos\theta,\,\sin\theta]^{\top}$, such that
$\mathcal{Q}(\theta)=q(\theta)q^{\top}(\theta)$.
Since $\|q(\theta)\|=1$, it follows immediately that
$\mathcal{Q}^{\top}(\theta)=\mathcal{Q}(\theta)$,
$\mathcal{Q}^{2}(\theta)=\mathcal{Q}(\theta)$, and
$\operatorname{rank}\mathcal{Q}(\theta)=1$. Thus,
$\mathcal{Q}(\theta)$ is the orthogonal projection onto
$\operatorname{span}\{q(\theta)\}$ and its spectrum satisfies
$\operatorname{spec}(\mathcal{Q}(\theta))=\{1,0\}$.
In particular,
\begin{align*}
    \mathcal{Q}(\theta)q(\theta)
    &=q(\theta),
    \\
    \mathcal{Q}(\theta)q_{\perp}(\theta)
    &=0,
\end{align*}
where
$q_{\perp}(\theta):=[-\sin\theta,\,\cos\theta]^{\top}$.
Hence, $q(\theta)$ and $q_{\perp}(\theta)$ are the eigenvectors associated
with the eigenvalues $1$ and $0$, respectively. Therefore,
$\mathcal{Q}(\theta)$ preserves the component aligned with the
short-circuit winding axis while eliminating its orthogonal component.

For the stator inter-turn fault, the spatial orientation of the
short-circuit winding is specified by $\gamma_{cc}$, with
$p\gamma_{cc}$ denoting the corresponding electrical orientation.
Consequently, $\mathcal{Q}(p\gamma_{cc})$ represents this orientation as
a directional projection acting on the $(d,q)$ quantities.

The spatial orientation $\gamma_{cc}$ is fixed with respect to the stator,
whereas the synchronous $(d,q)$ reference frame rotates during operation.
Let $\gamma(t)$ denote the mechanical angular position of the synchronous
$(d,q)$ frame with respect to the stationary stator frame. Its corresponding
electrical angle is $p\gamma(t)$ and satisfies
\begin{equation}
    p\dot{\gamma}(t)=\omega_s,
    \label{eq:synchronous_frame_angle}
\end{equation}
where $\omega_s$ is the synchronous electrical angular velocity introduced
in the previous section. For constant $\omega_s$ and $\gamma(0)=0$, this gives
$\gamma(t)=\omega_s t/p$.

Using $\mathcal{P}(\cdot)$ to account for this reference-frame rotation and
$\mathcal{Q}(p\gamma_{cc})$ to represent the spatial orientation of the
short-circuit winding, the fault-induced current in the synchronous $(d,q)$
frame is
\begin{equation}
    i_{cc,dq}
    =
    \frac{2\eta_{cc}}{3R_s}
    \mathcal{P}(-p\gamma)
    \mathcal{Q}(p\gamma_{cc})
    \mathcal{P}(p\gamma)
    V_{dq},
    \label{eq:icc_dq}
\end{equation}
where $i_{cc,dq}\in\mathbb{R}^{2}$ is the current contribution associated
with the short-circuited winding,
$V_{dq}=[V_{sd},\,V_{sq}]^{\top}\in\mathbb{R}^{2}$ is the stator-voltage
vector, and $R_s>0$ is the stator resistance.
Equation~\eqref{eq:icc_dq} separates the roles of the two fault parameters:
$\eta_{cc}$ determines the magnitude of the fault-induced current, whereas
$\gamma_{cc}$ determines its spatial orientation through
$\mathcal{Q}(p\gamma_{cc})$.

More precisely, define
\begin{equation*}
    \mathcal{G}(\gamma,\gamma_{cc},V_{dq})
    :=
    \frac{2}{3R_s}
    \mathcal{P}(-p\gamma)
    \mathcal{Q}(p\gamma_{cc})
    \mathcal{P}(p\gamma)
    V_{dq}
    \in\mathbb{R}^{2}.
\end{equation*}
Then, the short-circuit contribution admits the compact factorization
\begin{equation}
    i_{cc,dq}
    =
    \eta_{cc}
    \mathcal{G}(\gamma,\gamma_{cc},V_{dq}).
    \label{eq:fault_factorization}
\end{equation}
Thus, $\mathcal{G}(\gamma,\gamma_{cc},V_{dq})$ represents the
fault-induced current per unit severity. For a fixed spatial orientation
and known operating condition, $i_{cc,dq}$ therefore depends linearly on
the fault severity $\eta_{cc}$.

\begin{proposition}[Spatial periodicity of the fault signature]
\label{prop:fault_spatial_periodicity}
The orientation matrix satisfies
$\mathcal{Q}(\theta+\pi)=\mathcal{Q}(\theta)$. Consequently, the
fault-induced current per unit severity is periodic with respect to the
fault orientation, namely,
\begin{equation}
    \mathcal{G}
    \left(
        \gamma,
        \gamma_{cc}+\frac{\pi}{p},
        V_{dq}
    \right)
    =
    \mathcal{G}
    \left(
        \gamma,
        \gamma_{cc},
        V_{dq}
    \right).
    \label{eq:fault_signature_periodicity}
\end{equation}
Hence, the corresponding output signature also satisfies
$\psi(t;\gamma_{cc}+\pi/p)=\psi(t;\gamma_{cc})$.
\end{proposition}

\begin{proof}
Since $q(\theta+\pi)=-q(\theta)$ and
$\mathcal{Q}(\theta)=q(\theta)q^{\top}(\theta)$, it follows directly that
$\mathcal{Q}(\theta+\pi)=\mathcal{Q}(\theta)$. Setting
$\theta=p\gamma_{cc}$ in the definition of $\mathcal{G}$ gives
\eqref{eq:fault_signature_periodicity}. The corresponding periodicity of
$\psi(t;\gamma_{cc})$ follows from
\eqref{eq:fault_signature_vector}.
\end{proof}

\begin{remark}
\label{rem:distinct_fault_orientations}
Suppose that $n_e/(2p)\in\mathbb{N}$. By
Proposition~\ref{prop:fault_spatial_periodicity}, the admissible slot
orientations in \eqref{eq:fault_orientation_set} generate only
$N_{\gamma}=n_e/(2p)$ distinct first-harmonic fault signatures over the
fundamental interval $[0,\pi/p)$. In particular, for $n_e=36$ and $p=2$,
there are nine distinct orientations,
$\gamma_{cc}\in\{0^\circ,10^\circ,\ldots,80^\circ\}$.
\end{remark}

\subsection{Fault-Induced Output Signature}

Having characterized the fault-induced current $i_{cc,dq}$, we next determine
how this current appears in the measured output. Let
$i_{s,dq}=[i_{sd},\,i_{sq}]^{\top}\in\mathbb{R}^{2}$ denote the nominal
stator-current vector generated by the common-mode model of the previous
section. The corresponding stator current in the presence of the winding
fault is
\begin{equation}
    \begin{aligned}
        i_{s,dq}^{f}
        &=
        i_{s,dq}
        +
        i_{cc,dq}
        \\
        &=
        i_{s,dq}
        +
        \eta_{cc}
        \mathcal{G}(\gamma,\gamma_{cc},V_{dq}).
    \end{aligned}
    \label{eq:faulty_current}
\end{equation}
The healthy case is recovered immediately from $\eta_{cc}=0$, for which
$i_{cc,dq}=0$ and $i_{s,dq}^{f}=i_{s,dq}$. For $\eta_{cc}>0$, the deviation
from the nominal current is determined by the spatial orientation of the
short-circuit winding and the instantaneous electrical excitation.

Combining the nominal common-mode dynamics with the differential-mode fault
representation yields
\begin{equation}
    \begin{aligned}
        \dot{x}(t)
        &=
        f(x(t),u(t),T_l(t)),
        \\
        y^{f}(t)
        &=
        Cx(t)
        +
        D\,\eta_{cc}(t)
        \mathcal{G}\!\left(
            \gamma(t),\gamma_{cc},V_{dq}(t)
        \right).
    \end{aligned}
    \label{eq:fault_system}
\end{equation}
where
$D=[\,1~0;\,0~1;\,0~0\,]\in\mathbb{R}^{3\times2}$
embeds the fault-induced current into the measured output
$y^{f}=[i_{sd}^{f},\,i_{sq}^{f},\,\omega_r]^{\top}\in\mathbb{R}^{3}$.
Thus, the nominal electromechanical dynamics remain in the state equation,
whereas the differential-mode fault appears explicitly in the measured
output.

For subsequent estimation, it is useful to isolate the dependence of the
measured output on the unknown fault severity $\eta_{cc}(t)$. To this end,
define the fault-signature vector
\begin{equation}
    \psi(t;\gamma_{cc})
    :=
    D\,
    \mathcal{G}\!\left(
        \gamma(t),\gamma_{cc},V_{dq}(t)
    \right)
    \in\mathbb{R}^{3}.
    \label{eq:fault_signature_vector}
\end{equation}
The faulty output equation can then be written in the affine form
\begin{equation}
    y^{f}(t)
    =
    Cx(t)
    +
    \eta_{cc}(t)\psi(t;\gamma_{cc}).
    \label{eq:fault_output_affine}
\end{equation}
For a prescribed spatial orientation $\gamma_{cc}$ and known operating
condition, $\psi(t;\gamma_{cc})$ is known, whereas $\eta_{cc}(t)$ is the
unknown scalar degradation variable. Hence, although the complete motor
model is nonlinear, the differential-mode contribution is affine in the
fault severity. In particular, its instantaneous sensitivity satisfies
$\partial y^{f}(t)/\partial\eta_{cc}=\psi(t;\gamma_{cc})$.

From \eqref{eq:fault_output_affine}, the measured output $y^{f}(t)$ consists
of the nominal motor output $Cx(t)$ and the additional fault contribution
$\eta_{cc}(t)\psi(t;\gamma_{cc})$. To isolate the latter, define the fault
residual as $r_f(t):=y^{f}(t)-Cx(t)$. It follows that
\begin{equation}
    r_f(t)
    =
    \eta_{cc}(t)\psi(t;\gamma_{cc}).
    \label{eq:fault_residual_structure}
\end{equation}
Thus, $\eta_{cc}(t)=0$ implies $r_f(t)=0$. Conversely, whenever
$\psi(t;\gamma_{cc})\neq0$, $r_f(t)=0$ if and only if
$\eta_{cc}(t)=0$. Moreover, since $\eta_{cc}(t)\geq0$,
\begin{equation}
    \|r_f(t)\|
    =
    \eta_{cc}(t)
    \|\psi(t;\gamma_{cc})\|.
    \label{eq:residual_magnitude}
\end{equation}
Hence, for a fixed instantaneous fault signature, the residual magnitude
scales directly with the fault severity.

Equation~\eqref{eq:fault_residual_structure} therefore provides a direct
relation between the measurable effect of the winding fault and its unknown
severity $\eta_{cc}(t)$. The fault signature $\psi(t;\gamma_{cc})$ determines
how the spatial orientation of the short-circuit winding and the operating
condition shape this effect, whereas $\eta_{cc}(t)$ determines its magnitude.
Since $\eta_{cc}(t)$ is generally unknown and may evolve as the winding
deteriorates, it is estimated jointly with the motor states in the following
section and subsequently used as the health variable for prognosis and
reliability assessment.

\section{Augmented-State Particle Filtering}
\label{sec:particle_filtering}

The fault model in \eqref{eq:fault_output_affine} expresses the measured
output $y^{f}(t)$ in terms of the nominal motor state $x(t)$, the fault
severity $\eta_{cc}(t)$, and the fault signature
$\psi(t;\gamma_{cc})$. Since $\eta_{cc}(t)$ is not directly measurable,
the estimation problem is to reconstruct the electromechanical state
$x(t)\in\mathbb{R}^{5}$ while simultaneously estimating the fraction of
short-circuited turns $\eta_{cc}(t)\in[0,1]$ from the available noisy
measurements.

For the estimation stage, the spatial orientation $\gamma_{cc}$ of the
short-circuit winding is assumed to be fixed and known. Consequently,
$\psi(t;\gamma_{cc})$ is determined by the prescribed spatial orientation,
reference-frame angle, and known stator-voltage excitation,
while $\eta_{cc}(t)$ remains the
unknown time-varying degradation variable. This allows the estimator to
focus on the evolution of the fault severity. If $\gamma_{cc}$ is unknown, the estimation may instead be evaluated over
the distinct orientation set
$\Gamma_{cc}:=\{k2\pi/n_e:\,k=0,\ldots,n_e/(2p)-1\}$,
whenever $n_e/(2p)\in\mathbb{N}$, by
Proposition~\ref{prop:fault_spatial_periodicity}.

\subsection{Augmented Stochastic Model}

Let $T_s>0$ denote the sampling period. To obtain a discrete-time model for
recursive estimation, let $F_x(x_k,u_k,T_{l,k})$ denote the one-step
propagation of the nonlinear motor dynamics in
\eqref{eq:im_compact_dynamics} over the sampling interval $T_s$. The
discrete-time state model is
\begin{equation}
    x_{k+1}
    =
    F_x(x_k,u_k,T_{l,k})
    +
    w_{x,k},
    \label{eq:discrete_motor_model}
\end{equation}
where $x_k:=x(kT_s)$, $u_k:=u(kT_s)$, and
$T_{l,k}:=T_l(kT_s)$. The process disturbance is modeled as
$w_{x,k}\sim\mathcal{N}(0,Q_x)$, where
$Q_x\in\mathbb{R}^{5\times5}$ and $Q_x\succeq0$.

During online estimation, no long-term degradation law is imposed on
$\eta_{cc,k}$. Instead, its evolution between consecutive sampling instants is modeled
locally as a random walk,
\begin{equation*}
    \eta_{cc,k+1}
    =
    \eta_{cc,k}
    +
    w_{\eta,k},
\end{equation*}
where $w_{\eta,k}\sim\mathcal{N}(0,Q_{\eta})$ is the stochastic increment
and $Q_{\eta}>0$ controls the adaptability of the estimator to changes in
the winding degradation.
Since $\eta_{cc,k}$ represents the fraction of short-circuited turns, it must
satisfy $\eta_{cc,k}\in\mathcal{E}:=[0,1]$.
The unconstrained random walk, however, does not preserve this interval. 
We therefore define the projection
$\Pi_{\mathcal{E}}(z):=\min\{1,\max\{0,z\}\}$ and constrain the severity
transition as
\begin{equation}
    \eta_{cc,k+1}
    =
    \Pi_{\mathcal{E}}
    \left(
        \eta_{cc,k}+w_{\eta,k}
    \right).
    \label{eq:constrained_severity_transition}
\end{equation}
Thus, the random walk provides local adaptability without allowing the
estimated fraction of short-circuited turns to leave its physically
admissible interval. Importantly, \eqref{eq:constrained_severity_transition}
is an estimation model rather than a prescribed long-term degradation law;
future fault evolution is modeled separately in the prognosis stage.

To estimate the motor state and fault severity jointly, define the augmented
state
\begin{equation}
    z_k
    :=
    \begin{bmatrix}
        x_k\\
        \eta_{cc,k}
    \end{bmatrix}
    \in
    \mathcal{Z}
    :=
    \mathbb{R}^{5}\times[0,1].
    \label{eq:augmented_state}
\end{equation}
Combining \eqref{eq:discrete_motor_model},
\eqref{eq:constrained_severity_transition}, and the fault-dependent output
model \eqref{eq:fault_output_affine} gives the augmented stochastic system
\begin{equation}
    \begin{aligned}
        z_{k+1}
        &=
        \begin{bmatrix}
            F_x(x_k,u_k,T_{l,k})+w_{x,k}\\
            \Pi_{\mathcal{E}}(\eta_{cc,k}+w_{\eta,k})
        \end{bmatrix},
        \\
        y_k^{f}
        &=
        Cx_k+\eta_{cc,k}\psi_k(\gamma_{cc})+v_k.
    \end{aligned}
    \label{eq:augmented_stochastic_system}
\end{equation}
where
$w_k=[w_{x,k}^{\top},\,w_{\eta,k}]^{\top}
\sim\mathcal{N}(0,Q)$ with
$Q=\operatorname{diag}(Q_x,Q_{\eta})\in\mathbb{R}^{6\times6}$, while
$\psi_k(\gamma_{cc})
:=D\mathcal{G}(\gamma_k,\gamma_{cc},V_{dq,k})\in\mathbb{R}^{3}$.
The measurement noise satisfies
$v_k\sim\mathcal{N}(0,R)$, where
$R\in\mathbb{R}^{3\times3}$ and $R\succ0$.

Accordingly, the augmented state $z_k$ contains the five electromechanical
states governing the nominal motor dynamics together with the scalar
degradation variable $\eta_{cc,k}$ governing the fault contribution to the
measured output.

\subsection{Bayesian State--Fault Estimation}

Let $\mathcal{Y}_k:=\{y_1^{f},\ldots,y_k^{f}\}$ denote the measurement
history available at sampling instant $k$. Given $\mathcal{Y}_k$, the
estimation objective is to infer jointly the electromechanical state $x_k$
and the fault severity $\eta_{cc,k}$ through the posterior distribution
\begin{equation}
    p(z_k\mid\mathcal{Y}_k)
    =
    p(x_k,\eta_{cc,k}\mid\mathcal{Y}_k).
    \label{eq:joint_posterior}
\end{equation}
Thus, the motor state and the fraction of short-circuited turns are estimated
within a single Bayesian inference problem.

The stochastic model \eqref{eq:augmented_stochastic_system} defines the
state-transition density and measurement likelihood required for recursive
Bayesian estimation. Since the motor input $u_k$, load torque $T_{l,k}$,
and the operating quantities $\gamma_k$ and $V_{dq,k}$ entering the fault
signature $\psi_k(\gamma_{cc})$ are known at each sampling instant, their
conditioning is suppressed below for notational simplicity.
Given the posterior $p(z_{k-1}\mid\mathcal{Y}_{k-1})$, the Chapman--Kolmogorov
prediction gives
\begin{equation}
    p(z_k\mid\mathcal{Y}_{k-1})
    =
    \int_{\mathcal{Z}}
    p(z_k\mid z_{k-1})
    p(z_{k-1}\mid\mathcal{Y}_{k-1})
    \,\mbox{d}z_{k-1}.
    \label{eq:bayes_prediction}
\end{equation}
This step propagates the uncertainty in both $x_{k-1}$ and
$\eta_{cc,k-1}$ through the augmented stochastic dynamics.

Upon receiving the measurement $y_k^{f}$, the predicted density is corrected
using the Gaussian likelihood
$p(y_k^{f}\mid z_k)\propto
\exp\{-\frac{1}{2}
[y_k^{f}-Cx_k-\eta_{cc,k}\psi_k(\gamma_{cc})]^{\top}
R^{-1}
[y_k^{f}-Cx_k-\eta_{cc,k}\psi_k(\gamma_{cc})]\}$.
Bayes' rule then gives
\begin{equation}
    p(z_k\mid\mathcal{Y}_k)
    =
    \frac{
        p(y_k^{f}\mid z_k)
        p(z_k\mid\mathcal{Y}_{k-1})
    }{
        p(y_k^{f}\mid\mathcal{Y}_{k-1})
    }.
    \label{eq:bayes_update}
\end{equation}
The likelihood therefore assigns greater posterior probability to augmented
states whose predicted output is more consistent with the measured output
$y_k^{f}$.

Because the nonlinear motor dynamics and the fault-dependent measurement map
in \eqref{eq:augmented_stochastic_system} prevent, in general, closed-form
evaluation of the prediction and update equations, the posterior
$p(z_k\mid\mathcal{Y}_k)$ is approximated numerically using a particle
filter.

\subsection{Sequential Importance Sampling}

Since the posterior in \eqref{eq:joint_posterior} is generally unavailable
in closed form, it is approximated by $N_p$ weighted particles as
\begin{equation}
    p(z_k\mid\mathcal{Y}_k)
    \approx
    \sum_{i=1}^{N_p}
    w_k^{(i)}
    \delta
    \left(
        z_k-z_k^{(i)}
    \right),
    \label{eq:particle_density}
\end{equation}
where
$z_k^{(i)}=[(x_k^{(i)})^{\top},\,\eta_{cc,k}^{(i)}]^{\top}\in\mathcal{Z}$
is the $i$th augmented particle and $w_k^{(i)}\geq0$ is its normalized
importance weight, with $\sum_{i=1}^{N_p}w_k^{(i)}=1$.

Using the state-transition density as the importance proposal, each particle
is propagated according to
\begin{equation}
    \begin{aligned}
        x_k^{(i)}
        &=
        F_x
        \left(
            x_{k-1}^{(i)},
            u_{k-1},
            T_{l,k-1}
        \right)
        +
        w_{x,k-1}^{(i)},
        \\
        \eta_{cc,k}^{(i)}
        &=
        \Pi_{\mathcal{E}}
        \left(
            \eta_{cc,k-1}^{(i)}
            +
            w_{\eta,k-1}^{(i)}
        \right).
    \end{aligned}
    \label{eq:particle_prediction}
\end{equation}
Thus, each particle represents one possible realization of both the
electromechanical state and the physically admissible fault severity at
sampling instant $k$.

The predicted particles are then evaluated against the measured output
$y_k^{f}$. For the $i$th particle, define the innovation
$\nu_k^{(i)}
:=y_k^{f}-Cx_k^{(i)}
-\eta_{cc,k}^{(i)}\psi_k(\gamma_{cc})
\in\mathbb{R}^{3}$.
Since the measurement noise in \eqref{eq:augmented_stochastic_system}
satisfies $v_k\sim\mathcal{N}(0,R)$, the corresponding likelihood is
\begin{equation}
    p(y_k^{f}\mid z_k^{(i)})
    =
    \frac{
        \exp\!\left[
        -\frac{1}{2}
        (\nu_k^{(i)})^{\top}
        R^{-1}
        \nu_k^{(i)}
        \right]
    }{
        (2\pi)^{3/2}\det(R)^{1/2}
    }.
    \label{eq:particle_likelihood}
\end{equation}
Hence, particles producing smaller innovations relative to the measurement
covariance $R$ receive larger likelihood values.

Using the transition density as the proposal causes the proposal-dependent
terms in the general importance-weight recursion to cancel. The weights
therefore reduce to
\begin{equation}
    \begin{aligned}
        \widetilde w_k^{(i)}
        &=
        w_{k-1}^{(i)}
        p(y_k^{f}\mid z_k^{(i)}),
        \\
        w_k^{(i)}
        &=
        \frac{\widetilde w_k^{(i)}}
        {\displaystyle\sum_{j=1}^{N_p}\widetilde w_k^{(j)}}.
    \end{aligned}
    \label{eq:particle_weight_update}
\end{equation}
Consequently, the posterior particle population concentrates on augmented
states whose predicted stator currents and rotor speed are most consistent
with the measured motor response.

\subsection{Degeneracy and Resampling}

Repeated importance-weight updates may lead to particle degeneracy, whereby
the posterior probability becomes concentrated on only a small fraction of
the particle population. To monitor this loss of particle diversity, define
the effective sample size
\begin{equation}
    N_{\mathrm{eff},k}
    :=
    \left[
        \sum_{i=1}^{N_p}
        \left(w_k^{(i)}\right)^2
    \right]^{-1}.
    \label{eq:effective_sample_size}
\end{equation}
Since the weights are nonnegative and normalized,
$1\leq N_{\mathrm{eff},k}\leq N_p$. Equivalently, the normalized effective
sample size $\rho_{\mathrm{eff},k}:=N_{\mathrm{eff},k}/N_p$ satisfies
$1/N_p\leq\rho_{\mathrm{eff},k}\leq1$.
The upper bound is attained for uniform
weights $w_k^{(i)}=1/N_p$, whereas values approaching one indicate severe
concentration of the posterior mass on a small number of particles.

To prevent persistent weight degeneracy, resampling is triggered whenever
$N_{\mathrm{eff},k}<N_{\mathrm{th}}$, equivalently
$\rho_{\mathrm{eff},k}<\rho_{\mathrm{res}}$, where
$N_{\mathrm{th}}:=\rho_{\mathrm{res}}N_p$ and
$\rho_{\mathrm{res}}\in(0,1)$ is a prescribed threshold.
The particles are then resampled according to their normalized weights $w_k^{(i)}$, so that
particles with large posterior probability are preferentially replicated
while particles with negligible probability are discarded. Following
resampling, the particle weights are reset as
\begin{equation}
    w_k^{(i)}
    =
    \frac{1}{N_p},
    \qquad
    i=1,\ldots,N_p.
    \label{eq:resampled_weights}
\end{equation}
Resampling therefore redistributes the finite particle population toward
regions of the augmented state space $\mathcal{Z}$ having significant
posterior probability, thereby maintaining an effective representation of
both the motor state $x_k$ and the fault severity $\eta_{cc,k}$.

\subsection{Posterior State and Fault Estimates}

Prior to resampling, the minimum-mean-square-error estimates are given by
the conditional posterior means, which under the particle approximation
\eqref{eq:particle_density} are evaluated as
\begin{equation}
    \begin{aligned}
        \widehat{x}_k
        &=
        \sum_{i=1}^{N_p}
        w_k^{(i)}x_k^{(i)}
        \in\mathbb{R}^{5},
        \\
        \widehat{\eta}_{cc,k}
        &=
        \sum_{i=1}^{N_p}
        w_k^{(i)}\eta_{cc,k}^{(i)}
        \in[0,1].
    \end{aligned}
    \label{eq:pf_posterior_estimates}
\end{equation}
In particular, $\widehat{\eta}_{cc,k}$ provides the online estimate of the
fraction of short-circuited stator turns. Since the complete posterior
particle population is retained, the uncertainty associated with this
estimate can also be quantified. Its posterior variance is approximated by
\begin{equation}
    \widehat{\sigma}_{\eta,k}^{2}
    =
    \sum_{i=1}^{N_p}
    w_k^{(i)}
    \left(
        \eta_{cc,k}^{(i)}
        -
        \widehat{\eta}_{cc,k}
    \right)^2.
    \label{eq:pf_fault_variance}
\end{equation}
Accordingly, the posterior uncertainty band is
$[\eta_k^{-},\eta_k^{+}]$, where
$\eta_k^{-}:=\max\{0,\widehat{\eta}_{cc,k}-2\widehat{\sigma}_{\eta,k}\}$
and
$\eta_k^{+}:=\min\{1,\widehat{\eta}_{cc,k}+2\widehat{\sigma}_{\eta,k}\}$.
Thus, the particle filter provides both the estimated winding degradation
$\widehat{\eta}_{cc,k}$ and its associated posterior uncertainty
$\widehat{\sigma}_{\eta,k}$.

The sequence $\{\widehat{\eta}_{cc,k}\}$ constitutes the estimated
degradation history available for prognosis, while the posterior particles
retain the uncertainty in the current winding condition. These quantities
provide the common information source for the prognostic models developed
in the following section, where the estimated fault severity is propagated
forward to predict its future evolution and subsequent reliability.

\section{Fault Prognosis}
\label{sec:fault_prognosis}

The particle filter developed in the previous section provides the current
fault-severity estimate $\widehat{\eta}_{cc,k}$ together with its posterior
uncertainty. These quantities characterize the winding condition at sampling
instant $k$; reliability assessment, however, requires information about its
future evolution. The prognosis objective is therefore to predict the
severity $\eta_{cc,k+r}$ over a prescribed horizon
$r\in\{1,\ldots,H\}$, where $H\in\mathbb{N}_{>0}$.

For history-based prognosis, define the estimated degradation history
$\mathcal{H}_{k}^{\eta}
:=\{\widehat{\eta}_{cc,1},\ldots,\widehat{\eta}_{cc,k}\}$.
A deterministic prognostic model maps this information into the point
forecast
\begin{equation}
    \widehat{\eta}_{cc,k+r\mid k},
    \qquad
    r\in\{1,\ldots,H\},
    \label{eq:deterministic_severity_forecast}
\end{equation}
where $\widehat{\eta}_{cc,k+r\mid k}$ denotes the predicted fraction of
short-circuited turns at sampling instant $k+r$ using information available
through time $k$.

Probabilistic prognosis instead characterizes the uncertainty in future
degradation through
\begin{equation}
    p\!\left(
        \eta_{cc,k+r}\mid\mathcal{I}_k
    \right),
    \qquad
    r\in\{1,\ldots,H\},
    \label{eq:predictive_density_general}
\end{equation}
where $\mathcal{I}_k$ denotes the information available to the prognostic
model. For history-based models, $\mathcal{I}_k=\mathcal{H}_{k}^{\eta}$,
whereas the particle-based forecast uses the complete PF posterior
$p(z_k\mid\mathcal{Y}_k)$ in \eqref{eq:joint_posterior}. Thus, the latter
retains information that would be lost by reducing the current winding
condition to the point estimate $\widehat{\eta}_{cc,k}$ alone.

The four prognostic methods considered below differ in how they model the
future evolution of the same physical degradation variable $\eta_{cc}$.
The linear-trend and Holt models provide deterministic forecasts from the
estimated degradation history, whereas the Bayesian and particle-based
models additionally characterize prediction uncertainty. This common
formulation permits the effect of the prognostic model to be assessed
independently of the underlying fault-estimation procedure.

\subsection{Linear Trend Model}

The linear-trend model predicts future winding degradation by approximating
the recent evolution of the estimated fault severity with a straight line.
Let $M\leq k$ denote the number of most recent PF estimates used for this
local approximation and define
$\mathcal{J}_k:=\{k-M+1,\ldots,k\}$. Over this window, model the estimated
severity as
\begin{equation}
    \widehat{\eta}_{cc,j}
    =
    a_k+b_k j+\varepsilon_j,
    \qquad
    j\in\mathcal{J}_k,
    \label{eq:linear_fault_model}
\end{equation}
where $a_k\in\mathbb{R}$ is the local intercept,
$b_k\in\mathbb{R}$ is the local degradation rate, and $\varepsilon_j$
denotes the fitting residual.

The parameters are identified from the recent degradation history by
least squares:
\begin{equation}
    (\hat a_k,\hat b_k)
    =
    \underset{a,b\in\mathbb{R}}{\arg\min}
    \sum_{j=k-M+1}^{k}
    \left(
        \widehat{\eta}_{cc,j}-a-bj
    \right)^2.
    \label{eq:linear_ls}
\end{equation}
To express the solution compactly, define
\[
    \Phi_k
    :=
    \begin{bmatrix}
        1 & k-M+1\\
        \vdots & \vdots\\
        1 & k
    \end{bmatrix}
    \in\mathbb{R}^{M\times2},
    \qquad
    \mathbf{y}_k^{\eta}
    :=
    \begin{bmatrix}
        \widehat{\eta}_{cc,k-M+1}\\
        \vdots\\
        \widehat{\eta}_{cc,k}
    \end{bmatrix}
    \in\mathbb{R}^{M}.
\]
For $M\geq2$, the columns of $\Phi_k$ are linearly independent, and hence
the least-squares solution is uniquely given by
\begin{equation}
    \begin{bmatrix}
        \hat a_k\\
        \hat b_k
    \end{bmatrix}
    =
    \left(
        \Phi_k^{\top}\Phi_k
    \right)^{-1}
    \Phi_k^{\top}\mathbf{y}_k^{\eta}.
    \label{eq:linear_ls_solution}
\end{equation}

Extrapolating the fitted local trend beyond the current sampling instant
gives the unconstrained forecast
$\widetilde{\eta}_{cc,k+r\mid k}^{\mathrm{LT}}
:=\hat a_k+\hat b_k(k+r)$ for $r\in\{1,\ldots,H\}$. Since linear
extrapolation does not, in general, preserve the physical severity interval
$\mathcal{E}=[0,1]$, the final forecast is projected onto $\mathcal{E}$:
\begin{equation}
    \widehat{\eta}_{cc,k+r\mid k}^{\mathrm{LT}}
    =
    \Pi_{\mathcal{E}}
    \left(
        \widetilde{\eta}_{cc,k+r\mid k}^{\mathrm{LT}}
    \right),
    \qquad
    r\in\{1,\ldots,H\}.
    \label{eq:linear_prediction}
\end{equation}

The window length $M$ determines the balance between smoothing and local
adaptation: larger $M$ attenuates fluctuations in the PF estimates, whereas
smaller $M$ places greater emphasis on recent changes in the degradation
rate.

\subsection{Holt Exponential Smoothing}

The second deterministic prognostic model uses Holt exponential smoothing to
represent the estimated fault severity through a recursively updated level
and trend. Unlike the linear-trend model, which fits the most recent $M$ PF
estimates simultaneously, Holt smoothing updates these quantities whenever a
new estimate $\widehat{\eta}_{cc,k}$ becomes available.

Let $\ell_k$ and $d_k$ denote the estimated severity level and local
degradation trend at sampling instant $k$, respectively, initialized from
the available degradation history with $d_{k_0}\geq0$. Their recursive
updates are
\begin{equation}
    \begin{aligned}
        \ell_k
        &=
        \alpha\widehat{\eta}_{cc,k}
        +
        (1-\alpha)
        \left(
            \ell_{k-1}+d_{k-1}
        \right),
        \\
        d_k
        &=
        \max\left\{
            0,\,
            \beta
            \left(
                \ell_k-\ell_{k-1}
            \right)
            +
            (1-\beta)d_{k-1}
        \right\}.
    \end{aligned}
    \label{eq:holt_recursion}
\end{equation}
where $\alpha,\beta\in(0,1)$ are the level- and trend-smoothing parameters,
respectively, and $d_k\geq0$ enforces a nonnegative local degradation trend.
The parameter $\alpha$ controls how strongly the current PF estimate affects
the updated level, whereas $\beta$ determines how rapidly the estimated
degradation trend responds to changes in that level.

Propagating the current level with the estimated trend gives the
unconstrained $r$-step-ahead forecast
$\widetilde{\eta}_{cc,k+r\mid k}^{\mathrm{H}}
:=\ell_k+r d_k$ for $r\in\{1,\ldots,H\}$. As with the linear-trend model,
this extrapolation does not necessarily preserve the physical severity
interval $\mathcal{E}=[0,1]$. The Holt forecast is therefore defined as
\begin{equation}
    \widehat{\eta}_{cc,k+r\mid k}^{\mathrm{H}}
    =
    \Pi_{\mathcal{E}}
    \left(
        \widetilde{\eta}_{cc,k+r\mid k}^{\mathrm{H}}
    \right),
    \qquad
    r\in\{1,\ldots,H\}.
    \label{eq:holt_prediction}
\end{equation}

Holt smoothing therefore provides a recursive deterministic predictor that
retains only the current level $\ell_k$ and trend $d_k$. The parameters
$\alpha$ and $\beta$ determine the balance between smoothing the PF-based
degradation history and adapting to recent changes, while
\eqref{eq:holt_prediction} ensures that the resulting prediction remains
physically admissible.

\subsection{Bayesian Degradation Model}

The deterministic models above provide point forecasts but do not quantify
the uncertainty in future fault severity. To obtain a probabilistic forecast,
introduce the latent degradation process
\begin{equation}
    \widetilde{\eta}_{cc,k+1}
    =
    \widetilde{\eta}_{cc,k}
    +
    \mu
    +
    \xi_k,
    \label{eq:bayesian_degradation}
\end{equation}
where $\mu\in\mathbb{R}$ is the unknown mean degradation increment per
sampling interval and
$\xi_k\sim\mathcal{N}(0,\sigma_d^2)$, with $\sigma_d^2>0$, represents
stochastic degradation variability. The latent variable
$\widetilde{\eta}_{cc,k}\in\mathbb{R}$ is introduced to preserve Gaussian
tractability; physical admissibility is imposed subsequently on the
predictive distribution.

Rather than prescribing $\mu$, its value is inferred locally from the
PF-based degradation history $\mathcal{H}_k^\eta$. Let $M_B\leq k-1$
denote the number of recent increments used for inference and define
$\mathcal{J}_k^{\mathrm{B}}:=\{k-M_B+1,\ldots,k\}$.
Assign the Gaussian prior
$\mu\sim\mathcal{N}(\mu_0,\sigma_{\mu,0}^2)$, where
$\sigma_{\mu,0}^2>0$, and define
$\Delta\widehat{\eta}_{cc,j}
:=\widehat{\eta}_{cc,j}-\widehat{\eta}_{cc,j-1}$ for
$j\in\mathcal{J}_k^{\mathrm{B}}$.
Under \eqref{eq:bayesian_degradation},
$\Delta\widehat{\eta}_{cc,j}\mid\mu
\sim\mathcal{N}(\mu,\sigma_d^2)$.

For known $\sigma_d^2$, Gaussian conjugacy gives
\begin{equation}
    \mu\mid\mathcal{H}_k^\eta
    \sim
    \mathcal{N}
    \left(
        \mu_k,\sigma_{\mu,k}^2
    \right),
    \qquad
    \begin{cases}
        \displaystyle
        \sigma_{\mu,k}^2
        =
        \left(
            \frac{1}{\sigma_{\mu,0}^2}
            +
            \frac{M_B}{\sigma_d^2}
        \right)^{-1},
        \\[3ex]
        \displaystyle
        \mu_k
        =
        \sigma_{\mu,k}^2
        \left[
            \frac{\mu_0}{\sigma_{\mu,0}^2}
            +
            \frac{1}{\sigma_d^2}
            \sum\nolimits_{j\in\mathcal{J}_k^{\mathrm{B}}}
            \Delta\widehat{\eta}_{cc,j}
        \right].
    \end{cases}
    \label{eq:degradation_rate_posterior}
\end{equation}
Thus, $\mu_k$ estimates the mean degradation increment, while
$\sigma_{\mu,k}^2$ quantifies the remaining uncertainty in that rate.

Starting from the current PF estimate $\widehat{\eta}_{cc,k}$ and iterating
\eqref{eq:bayesian_degradation} over $r$ future sampling intervals gives
$\widetilde{\eta}_{cc,k+r}
=
\widehat{\eta}_{cc,k}
+r\mu+\sum_{j=0}^{r-1}\xi_{k+j}$.
Marginalizing over the posterior uncertainty in $\mu$ therefore yields
\begin{equation}
    \widetilde{\eta}_{cc,k+r}
    \mid
    \mathcal{H}_k^\eta
    \sim
    \mathcal{N}
    \left(
        m_{k,r},
        P_{k,r}
    \right),
    \qquad
    \begin{cases}
        \displaystyle
        m_{k,r}
        =
        \widehat{\eta}_{cc,k}
        +
        r\mu_k,
        \\[2ex]
        \displaystyle
        P_{k,r}
        =
        r\sigma_d^2
        +
        r^2\sigma_{\mu,k}^2.
    \end{cases}
    \label{eq:bayesian_predictive_distribution}
\end{equation}
The two terms in $P_{k,r}$ have distinct origins:
$r\sigma_d^2$ is the accumulated stochastic degradation uncertainty, whereas
$r^2\sigma_{\mu,k}^2$ results from uncertainty in the inferred degradation
rate. The latter therefore becomes increasingly influential as the prediction
horizon increases.

The Gaussian model above is defined on $\mathbb{R}$, whereas the physical
fault severity must lie in $\mathcal{E}=[0,1]$. For deterministic comparison,
we use the projected point forecast
\begin{equation}
    \widehat{\eta}_{cc,k+r\mid k}^{\mathrm{B}}
    =
    \Pi_{\mathcal{E}}(m_{k,r}),
    \qquad
    r\in\{1,\ldots,H\}.
    \label{eq:bayesian_point_prediction}
\end{equation}
For probabilistic prognosis, physical admissibility is instead enforced on
the complete predictive distribution. Let $\phi(\cdot)$ and $\Phi(\cdot)$
denote the standard Gaussian density and distribution functions,
respectively. Conditioning the Gaussian predictive law in
\eqref{eq:bayesian_predictive_distribution} on
$\widetilde{\eta}_{cc,k+r}\in[0,1]$ gives
\begin{equation}
    p_{\mathcal{E}}
    \left(
        \eta\mid\mathcal{H}_{k}^{\eta}
    \right)
    =
    \frac{
        \displaystyle
        \frac{1}{\sqrt{P_{k,r}}}
        \phi
        \left(
            \frac{\eta-m_{k,r}}{\sqrt{P_{k,r}}}
        \right)
    }{
        \displaystyle
        \Phi
        \left(
            \frac{1-m_{k,r}}{\sqrt{P_{k,r}}}
        \right)
        -
        \Phi
        \left(
            \frac{-m_{k,r}}{\sqrt{P_{k,r}}}
        \right)
    },
    \qquad
    \eta\in[0,1].
    \label{eq:truncated_bayesian_predictive}
\end{equation}
Hence, the Bayesian model provides both a physically admissible point
forecast and a horizon-dependent predictive distribution. The latter retains
uncertainty in both the degradation process and its inferred rate and will
be used directly in the probabilistic reliability assessment.

\subsection{Particle-Based Forecast}

The Bayesian model characterizes future degradation through a parametric
Gaussian distribution initialized from the point estimate
$\widehat{\eta}_{cc,k}$. The particle-based forecast instead propagates an
empirical representation of the current severity uncertainty.

Let $\{\eta_{cc,k}^{(i)},w_k^{(i)}\}_{i=1}^{N_p}$ denote the severity
particles and corresponding normalized weights at the prognosis instant.
When the augmented-state PF posterior is available, these are the severity
components and weights in \eqref{eq:particle_density}. The locally inferred
degradation increment $\mu_k$ in \eqref{eq:degradation_rate_posterior} is
then used for forward propagation.
Each severity particle is
then propagated over the prediction horizon according to
\begin{equation}
    \eta_{cc,k+r}^{(i)}
    =
    \Pi_{\mathcal{E}}
    \left(
        \eta_{cc,k+r-1}^{(i)}
        +
        \mu_k
        +
        \xi_{k+r-1}^{(i)}
    \right),
    \qquad
    r\in\{1,\ldots,H\},
    \label{eq:particle_forecast_propagation}
\end{equation}
where
$\xi_{k+r-1}^{(i)}\sim\mathcal{N}(0,\sigma_d^2)$ independently across
particles and prediction steps. The projection onto
$\mathcal{E}=[0,1]$ ensures that every predicted trajectory remains
physically admissible.

For each prediction step $r$, the propagated particles define the empirical
predictive distribution
\begin{equation}
    p
    \left(
        \eta_{cc,k+r}\mid\mathcal{Y}_k
    \right)
    \approx
    \sum_{i=1}^{N_p}
        w_k^{(i)}
        \delta
        \left(
            \eta_{cc,k+r}
            -
            \eta_{cc,k+r}^{(i)}
        \right).
    \label{eq:particle_predictive_distribution}
\end{equation}
Hence, the forecast propagates the current fault-severity uncertainty and
the stochastic variability of future degradation.

The corresponding point forecast and predictive variance are
\begin{equation}
    \widehat{\eta}_{cc,k+r\mid k}^{\mathrm{PF}}
    =
    \sum_{i=1}^{N_p}
        w_k^{(i)}
        \eta_{cc,k+r}^{(i)},
    \qquad
    P_{\eta,k+r}^{\mathrm{PF}}
    =
    \sum_{i=1}^{N_p}
        w_k^{(i)}
        \left(
            \eta_{cc,k+r}^{(i)}
            -
            \widehat{\eta}_{cc,k+r\mid k}^{\mathrm{PF}}
        \right)^2.
    \label{eq:particle_prediction_statistics}
\end{equation}
Unlike the Gaussian Bayesian model, the particle forecast imposes no
parametric form on the predictive distribution and can therefore represent
non-Gaussian uncertainty generated by the posterior distribution, stochastic
propagation, and projection onto the physical severity interval.

The empirical distribution also provides fault-risk quantities directly.
For a prescribed severity threshold $\eta_{\mathrm{crit}}\in(0,1)$, the
pointwise exceedance probability is approximated by
\begin{equation}
    \mathbb{P}
    \left(
        \eta_{cc,k+r}\geq\eta_{\mathrm{th}}
        \mid\mathcal{Y}_k
    \right)
    \approx
    \sum_{i=1}^{N_p}
        w_k^{(i)}
        \mathbf{1}_{[\eta_{\mathrm{th}},1]}
        \left(
            \eta_{cc,k+r}^{(i)}
        \right).
    \label{eq:particle_exceedance_probability}
\end{equation}
Prediction intervals are obtained analogously from weighted quantiles of
the empirical distribution. These quantities will be used directly in the
probabilistic reliability assessment of the following section.

The four prognostic models therefore provide complementary descriptions of
the same physical degradation variable $\eta_{cc}$. The linear-trend and
Holt models provide deterministic forecasts, whereas the Bayesian and
particle-based models additionally characterize prediction uncertainty.
All four methods use the same PF-based health information and prediction
horizon, enabling their forecasting performance and resulting reliability
estimates to be compared on a common basis.

\section{Reliability Assessment}
\label{sec:reliability}

The prognosis methods of Section~\ref{sec:fault_prognosis} characterize the
future stator-fault severity through point forecasts
$\widehat{\eta}_{cc,k+r\mid k}^{q}$, where
$q\in\{\mathrm{LT},\mathrm{H},\mathrm{B},\mathrm{PF}\}$ and
$r\in\{1,\ldots,H\}$, or through the predictive distribution
$p(\eta_{cc,k+r}\mid\mathcal{I}_k)$ when probabilistic information is
retained. Reliability assessment uses these predictions to quantify whether
the winding degradation remains below a prescribed failure condition.

Since $\eta_{cc}\in[0,1]$ represents the fraction of short-circuited stator
turns, let $\eta_{\mathrm{crit}}\in(0,1)$ denote the critical severity at
which failure is declared. For a degradation trajectory evolving from
sampling instant $k$, define the first-passage time
\begin{equation}
    \tau_k
    :=
    \inf
    \left\{
        r\in\mathbb{N}_{>0}:
        \eta_{cc,k+r}
        \geq
        \eta_{\mathrm{crit}}
    \right\},
    \label{eq:first_passage_time}
\end{equation}
with $\inf\varnothing=\infty$. Thus, $\tau_k$ is the number of future
sampling intervals until the degradation first reaches the prescribed
failure region, while the corresponding physical first-passage time is
$T_s\tau_k$.

The conditional reliability over the next $r$ sampling intervals is defined as
\begin{equation}
    \begin{aligned}
        R_k(r)
        &:=
        \mathbb{P}
        \left(
            \tau_k>r
            \mid
            \mathcal{I}_k
        \right)
        \\
        &=
        \mathbb{P}
        \left(
            \eta_{cc,k+j}<\eta_{\mathrm{crit}},
            \quad
            \forall j\in\{1,\ldots,r\}
            \,\middle|\,
            \mathcal{I}_k
        \right)
        \\
        &=
        \mathbb{P}
        \left(
            \max_{1\leq j\leq r}
            \eta_{cc,k+j}<\eta_{\mathrm{crit}}
            \,\middle|\,
            \mathcal{I}_k
        \right),
        \qquad
        r\in\{1,\ldots,H\}.
    \end{aligned}
    \label{eq:conditional_reliability}
\end{equation}
Hence, reliability is a first-passage quantity: the degradation must remain
below $\eta_{\mathrm{crit}}$ throughout the prediction interval, rather than
only at its terminal instant.

Based on this common definition, four reliability formulations are considered
below: deterministic threshold crossing, Weibull lifetime reliability,
Bayesian degradation reliability, and degradation-dependent hazard
reliability.

\subsection{Threshold-Crossing Reliability}

The threshold-crossing formulation maps each point forecast
$\widehat{\eta}_{cc,k+r\mid k}^{q}$ directly into a predicted failure time.
For prognostic method $q$, define
\begin{equation}
    \widehat{\tau}_{k}^{\,q}
    :=
    \inf
    \left\{
        r\in\{1,\ldots,H\}:
        \widehat{\eta}_{cc,k+r\mid k}^{q}
        \geq
        \eta_{\mathrm{crit}}
    \right\},
    \label{eq:deterministic_failure_time}
\end{equation}
with $\inf\varnothing=\infty$. Thus,
$\widehat{\tau}_{k}^{\,q}$ is the first prediction step at which the
forecast severity reaches the critical threshold. If no crossing occurs
within the available horizon, only
$\widehat{\tau}_{k}^{\,q}>H$ can be inferred.

The corresponding threshold-based reliability indicator is
\begin{equation}
    R_{k}^{\mathrm{th},q}(r)
    :=
    \mathbf{1}_{\{\widehat{\tau}_{k}^{\,q}>r\}}
    =
    \mathbf{1}_{\left\{
        \displaystyle
        \max_{1\leq j\leq r}
        \widehat{\eta}_{cc,k+j\mid k}^{q}
        <
        \eta_{\mathrm{crit}}
    \right\}},
    \qquad
    r\in\{1,\ldots,H\},
    \label{eq:threshold_reliability}
\end{equation}
where $\mathbf{1}_{\{\cdot\}}$ denotes the indicator function.
Hence, the deterministic forecast is regarded as reliable through step $r$
only if its entire predicted trajectory remains below
$\eta_{\mathrm{crit}}$ over that interval. This is the point-forecast
counterpart of the first-passage reliability definition in
\eqref{eq:conditional_reliability}.

Let $T_p>0$ denote the prognosis sampling interval. Whenever
$\widehat{\tau}_{k}^{\,q}\leq H$, the corresponding remaining useful life
estimate is
\begin{equation}
    \widehat{\mathrm{RUL}}_{k}^{\,q}
    =
    T_p\widehat{\tau}_{k}^{\,q},
    \label{eq:threshold_rul}
\end{equation}
If no threshold crossing occurs within the prediction horizon, a finite RUL
cannot be determined from the available forecast; only
$\mathrm{RUL}_{k}^{\,q}>T_pH$. can be inferred.

The threshold-crossing formulation therefore provides a direct mapping from
each prognostic trajectory to a predicted failure time and RUL. However,
$R_{k}^{\mathrm{th},q}(r)$ is binary and does not account for uncertainty in
either the future degradation or its threshold-crossing time.

\subsection{Weibull Reliability Model}

Unlike the threshold-crossing formulation, which derives reliability from
the predicted fault-severity trajectory, the Weibull model describes the
motor lifetime statistically. Let $T_f\in\mathbb{R}_{\geq0}$ denote the
time to failure and assume a two-parameter Weibull distribution with scale
$\lambda_W>0$ and shape $\beta_W>0$. Its survival function and hazard rate
are
\begin{equation}
    S_W(t)
    =
    \mathbb{P}(T_f>t)
    =
    \exp
    \left[
        -
        \left(
            \frac{t}{\lambda_W}
        \right)^{\beta_W}
    \right],
    \qquad
    h_W(t)
    =
    \frac{\beta_W}{\lambda_W}
    \left(
        \frac{t}{\lambda_W}
    \right)^{\beta_W-1},
    \quad t\geq0.
    \label{eq:weibull_model}
\end{equation}
The shape parameter determines the evolution of the failure rate:
$\beta_W<1$, $\beta_W=1$, and $\beta_W>1$ correspond to decreasing,
constant, and increasing hazards, respectively.

At prognosis instant $k$, let $t_k$ denote the corresponding operating time.
Reliability over the next $r$ prognosis intervals is conditioned on survival
up to $t_k$. Accordingly,
\begin{equation}
    \begin{aligned}
        R_k^{W}(r)
        &:=
        \mathbb{P}
        \left(
            T_f>t_k+rT_s
            \,\middle|\,
            T_f>t_k
        \right)
        \\
        &=
        \frac{S_W(t_k+rT_s)}{S_W(t_k)}
        \\
        &=
        \exp
        \left\{
            -
            \left[
                \left(
                    \frac{t_k+rT_s}{\lambda_W}
                \right)^{\beta_W}
                -
                \left(
                    \frac{t_k}{\lambda_W}
                \right)^{\beta_W}
            \right]
        \right\},
        \qquad
        r\in\{1,\ldots,H\}.
    \end{aligned}
    \label{eq:weibull_reliability}
\end{equation}
Equivalently,
$R_k^{W}(r)
=
\exp[-\int_{t_k}^{t_k+rT_p}h_W(\tau)\,\mbox{d}\tau]$,
showing that reliability is determined by the failure risk accumulated over
the future interval of interest.

The parameters $(\lambda_W,\beta_W)$ must be identified from appropriate
lifetime or failure-time data representative of the considered operating
conditions. In particular, they are not inferred from the instantaneous
fault estimate $\widehat{\eta}_{cc,k}$. The Weibull model therefore serves
as a population-level lifetime benchmark, whereas the other reliability
formulations considered here exploit the predicted evolution of the
physically defined fault severity $\eta_{cc}$.

\subsection{Bayesian Reliability}

The Bayesian degradation model provides a predictive distribution of the
future fault severity and therefore allows degradation uncertainty to enter
the reliability assessment directly. At prediction step $r$, define the
pointwise reliability as
\begin{equation}
    \widetilde{R}_{k}^{B}(r)
    :=
    \mathbb{P}
    \left(
        \eta_{cc,k+r}<\eta_{\mathrm{crit}}
        \,\middle|\,
        \mathcal{H}_{k}^{\eta}
    \right)
    =
    \int_{0}^{\eta_{\mathrm{crit}}}
        p_{\mathcal{E}}
        \left(
            \eta\mid\mathcal{H}_{k}^{\eta}
        \right)
    \,\mbox{d}\eta.
    \label{eq:bayesian_pointwise_reliability}
\end{equation}
Using the truncated Gaussian predictive density in
\eqref{eq:truncated_bayesian_predictive}, this probability becomes
\begin{equation}
    \widetilde{R}_{k}^{B}(r)
    =
    \frac{
        \displaystyle
        \Phi
        \left(
            \frac{\eta_{\mathrm{crit}}-m_{k,r}}
                 {\sqrt{P_{k,r}}}
        \right)
        -
        \Phi
        \left(
            \frac{-m_{k,r}}
                 {\sqrt{P_{k,r}}}
        \right)
    }{
        \displaystyle
        \Phi
        \left(
            \frac{1-m_{k,r}}
                 {\sqrt{P_{k,r}}}
        \right)
        -
        \Phi
        \left(
            \frac{-m_{k,r}}
                 {\sqrt{P_{k,r}}}
        \right)
    },
    \label{eq:truncated_bayesian_reliability}
\end{equation}
where $\Phi(\cdot)$ denotes the standard Gaussian cumulative distribution
function.

The quantity $\widetilde{R}_{k}^{B}(r)$ is marginal: it evaluates whether
the severity is below $\eta_{\mathrm{crit}}$ at time $k+r$ but does not
exclude an earlier threshold crossing. First-passage reliability instead
requires the complete predicted trajectory to remain below the critical
severity. Accordingly,
\begin{equation}
    \begin{aligned}
        R_k^{B}(r)
        &:=
        \mathbb{P}
        \left(
            \eta_{cc,k+j}<\eta_{\mathrm{crit}},
            \quad
            \forall j\in\{1,\ldots,r\}
            \,\middle|\,
            \mathcal{H}_{k}^{\eta}
        \right)
        \\
        &=
        \mathbb{P}
        \left(
            \max_{1\leq j\leq r}
            \eta_{cc,k+j}<\eta_{\mathrm{crit}}
            \,\middle|\,
            \mathcal{H}_{k}^{\eta}
        \right).
    \end{aligned}
    \label{eq:bayesian_path_reliability}
\end{equation}
Defining
$\boldsymbol{\eta}_{k,r}
:=[\eta_{cc,k+1},\ldots,\eta_{cc,k+r}]^{\top}$, the same probability can
be expressed as
\begin{equation}
    R_k^{B}(r)
    =
    \int_{[0,\eta_{\mathrm{crit}})^r}
        p_{\mathcal{E}}
        \left(
            \boldsymbol{\eta}_{k,r}
            \mid
            \mathcal{H}_{k}^{\eta}
        \right)
    \,\mbox{d}\boldsymbol{\eta}_{k,r}.
    \label{eq:bayesian_joint_path_reliability}
\end{equation}
Thus, first-passage reliability depends on the joint predictive distribution
of the future degradation path, rather than only on the marginal distribution
at step $r$.

For the particle-based forecast, this path probability is evaluated directly
from the propagated trajectories as
\begin{equation}
    R_k^{\mathrm{PF}}(r)
    \approx
    \sum_{i=1}^{N_p}
        w_k^{(i)}
        \mathbf{1}_{\left\{
            \displaystyle
            \max_{1\leq j\leq r}
            \eta_{cc,k+j}^{(i)}
            <
            \eta_{\mathrm{crit}}
        \right\}}.
    \label{eq:particle_path_reliability}
\end{equation}
Hence, $R_k^{\mathrm{PF}}(r)$ is the total weight of predicted severity
trajectories that avoid the failure region through step $r$, with
$w_k^{(i)}=1/N_p$ for an equally weighted prognostic ensemble.

The distinction is therefore precise:
$\widetilde{R}_{k}^{B}(r)$ measures survival at a particular future instant,
whereas $R_k^{B}(r)$ and $R_k^{\mathrm{PF}}(r)$ measure survival over the
entire prediction interval. The latter are consistent with the first-passage
reliability definition in \eqref{eq:conditional_reliability}.

\subsection{Degradation-Dependent Hazard Reliability}

The preceding formulations describe reliability through threshold crossing
or lifetime statistics. A complementary approach is to define the
instantaneous failure rate directly from the predicted fault severity. For
each prognostic method
$q\in\{\mathrm{LT},\mathrm{H},\mathrm{B},\mathrm{PF}\}$, define the
normalized predicted severity
\begin{equation}
    \rho_{k+r\mid k}^{q}
    :=
    \frac{
        \widehat{\eta}_{cc,k+r\mid k}^{q}
    }{
        \eta_{\mathrm{crit}}
    },
    \qquad
    r\in\{1,\ldots,H\}.
    \label{eq:normalized_predicted_severity}
\end{equation}
Thus, $\rho_{k+r\mid k}^{q}=1$ corresponds to the prescribed critical
severity. To map the predicted winding condition into a failure rate, define
the degradation-dependent hazard
\begin{equation}
    h_k^{q}(r)
    :=
    h_0
    \exp
    \left(
        \alpha\rho_{k+r\mid k}^{q}
    \right)
    =
    h_0
    \exp
    \left(
        \frac{
            \alpha\widehat{\eta}_{cc,k+r\mid k}^{q}
        }{
            \eta_{\mathrm{crit}}
        }
    \right),
    \label{eq:degradation_hazard}
\end{equation}
where $h_0>0$ is the baseline hazard rate and $\alpha>0$ determines the
sensitivity of the failure rate to the normalized fault severity.

The cumulative hazard and corresponding reliability over the next $r$
sampling intervals are
\begin{equation}
    \begin{aligned}
        \Lambda_k^{q}(r)
        &=
        T_p
        \sum_{j=1}^{r}
        h_k^{q}(j),
        \\
        R_k^{\mathrm{H},q}(r)
        &=
        \exp\left(-\Lambda_k^{q}(r)\right).
    \end{aligned}
    \label{eq:hazard_reliability}
\end{equation}
Substituting \eqref{eq:degradation_hazard} gives
\begin{equation}
    R_k^{\mathrm{H},q}(r)
    =
    \exp
    \left[
        -
        T_p h_0
        \sum_{j=1}^{r}
        \exp
        \left(
            \frac{
                \alpha
                \widehat{\eta}_{cc,k+j\mid k}^{q}
            }{
                \eta_{\mathrm{crit}}
            }
        \right)
    \right].
    \label{eq:hazard_reliability_explicit}
\end{equation}
Hence, larger predicted fault severities increase the instantaneous hazard
and accelerate the decay of reliability.

When a probabilistic degradation forecast is available, the uncertainty in
the future fault severity can be incorporated by averaging the
severity-dependent hazard over its predictive distribution. Specifically,
\begin{equation}
    \overline{h}_k(r)
    :=
    \mathbb{E}
    \left[
        h_0
        \exp
        \left(
            \frac{\alpha}{\eta_{\mathrm{crit}}}
            \eta_{cc,k+r}
        \right)
        \,\middle|\,
        \mathcal{I}_k
    \right].
    \label{eq:expected_hazard}
\end{equation}
Unlike \eqref{eq:degradation_hazard}, which evaluates the hazard along a
point forecast, \eqref{eq:expected_hazard} accounts for the complete
predictive distribution of the future fault severity.

For the Bayesian degradation model, let
$c:=\alpha/\eta_{\mathrm{crit}}$ and define the normalization factor
\[
    Z_{k,r}
    :=
    \Phi
    \left(
        \frac{1-m_{k,r}}{\sqrt{P_{k,r}}}
    \right)
    -
    \Phi
    \left(
        \frac{-m_{k,r}}{\sqrt{P_{k,r}}}
    \right).
\]
Using the truncated Gaussian predictive density in
\eqref{eq:truncated_bayesian_predictive}, the uncertainty-averaged hazard is
\begin{equation}
    \begin{aligned}
        \overline{h}_k^{B}(r)
        &=
        h_0
        \exp
        \left(
            c m_{k,r}
            +
            \frac{c^2P_{k,r}}{2}
        \right)
        \\
        &\quad\times
        \frac{
            \displaystyle
            \Phi
            \left(
                \frac{
                    1-m_{k,r}-cP_{k,r}
                }{
                    \sqrt{P_{k,r}}
                }
            \right)
            -
            \Phi
            \left(
                \frac{
                    -m_{k,r}-cP_{k,r}
                }{
                    \sqrt{P_{k,r}}
                }
            \right)
        }{
            Z_{k,r}
        }.
    \end{aligned}
    \label{eq:bayesian_expected_hazard}
\end{equation}
Thus, the Bayesian hazard incorporates both the predicted degradation level
and its uncertainty while respecting the physical constraint
$\eta_{cc}\in[0,1]$.

For the particle-based forecast, the same expectation is evaluated directly
from the propagated severity particles as
\begin{equation}
    \overline{h}_k^{\mathrm{PF}}(r)
    \approx
    h_0
    \sum_{i=1}^{N_p}
        w_k^{(i)}
        \exp
        \left(
            \frac{
                \alpha\eta_{cc,k+r}^{(i)}
            }{
                \eta_{\mathrm{crit}}
            }
        \right).
    \label{eq:particle_expected_hazard}
\end{equation}
This approximation retains the nonparametric predictive distribution
represented by the propagated particles, with $w_k^{(i)}=1/N_p$ for an
equally weighted prognostic ensemble.

For $q\in\{\mathrm{B},\mathrm{PF}\}$, the uncertainty-aware cumulative
hazard and reliability are then
\begin{equation}
    \begin{aligned}
        \overline{\Lambda}_k^{q}(r)
        &=
        T_p
        \sum_{j=1}^{r}
        \overline{h}_k^{q}(j),
        \\
        \overline{R}_k^{\mathrm{H},q}(r)
        &=
        \exp
        \left(
            -\overline{\Lambda}_k^{q}(r)
        \right).
    \end{aligned}
    \label{eq:uncertain_hazard_reliability}
\end{equation}
Therefore, $R_k^{\mathrm{H},q}(r)$ evaluates the accumulated failure risk
along a point forecast, whereas
$\overline{R}_k^{\mathrm{H},q}(r)$ incorporates uncertainty in the predicted
winding degradation.

The reliability formulations considered in this section provide
complementary interpretations of the prognostic information.
Threshold-crossing reliability converts a predicted severity trajectory into
a first-passage time and remaining useful life, while the Weibull model
provides an independently identified statistical lifetime benchmark.
Bayesian and particle-based reliability evaluate the probability of avoiding
the critical severity, whereas degradation-dependent hazard reliability maps
the evolving fault severity into accumulated failure risk.

All severity-dependent reliability measures originate from the same
particle-filter-based fault estimate and subsequent degradation forecasts.
This common foundation allows the effects of prognosis and reliability
modeling to be examined separately and completes the connection from online
fault estimation to degradation prognosis and reliability assessment.

\section{Numerical Results}
\label{sec:numerical_results}

This section evaluates the complete estimation--prognosis--reliability
framework. The nonlinear induction-motor dynamics are propagated using
fourth-order Runge--Kutta integration with sampling period
$T_s=5\times10^{-4}$~s over $10$~s. The motor parameters are
$L_m=0.258$~H, $L_s=L_r=0.274$~H, $R_s=4.85~\Omega$,
$R_r=3.805~\Omega$, $J=0.031$~kg\,m$^2$, $p=2$,
$\omega_s=100\pi$~rad/s, and $T_l=1$~N\,m. The applied voltages are
$V_{sd}(t)=220+20\sin(0.8\pi t)$ and
$V_{sq}(t)=20\sin(1.4\pi t)$.

Unless otherwise stated, the stator contains $n_e=36$ slots and the fault
orientation is $\gamma_{cc}=4(2\pi/n_e)=40^\circ$. The progressive
inter-turn fault begins at $t_f=2$~s and is generated as
$\eta_{cc}(t)=0.01+0.008(t-t_f)+0.001(t-t_f)^2$ for $t\geq t_f$, with
$\eta_{cc}(t)=0$ before fault occurrence and projection onto $[0,1]$.
The measurement-noise standard deviations are
$[0.8,\,0.8,\,1.0]^\top$, while the particle filter uses $N_p=1000$
particles, $N_{\mathrm{th}}=0.5N_p$,
$Q_x^{1/2}=\operatorname{diag}
\{0.03,0.03,4\times10^{-4},4\times10^{-4},0.05\}$, and
$Q_\eta^{1/2}=3\times10^{-4}$.

\subsection{Fault Estimation and Output Reconstruction}

\begin{figure}[t]
    \centering
    \subfloat[Fault severity.\label{fig:fig1a}]
    {\includegraphics[width=0.25\linewidth]{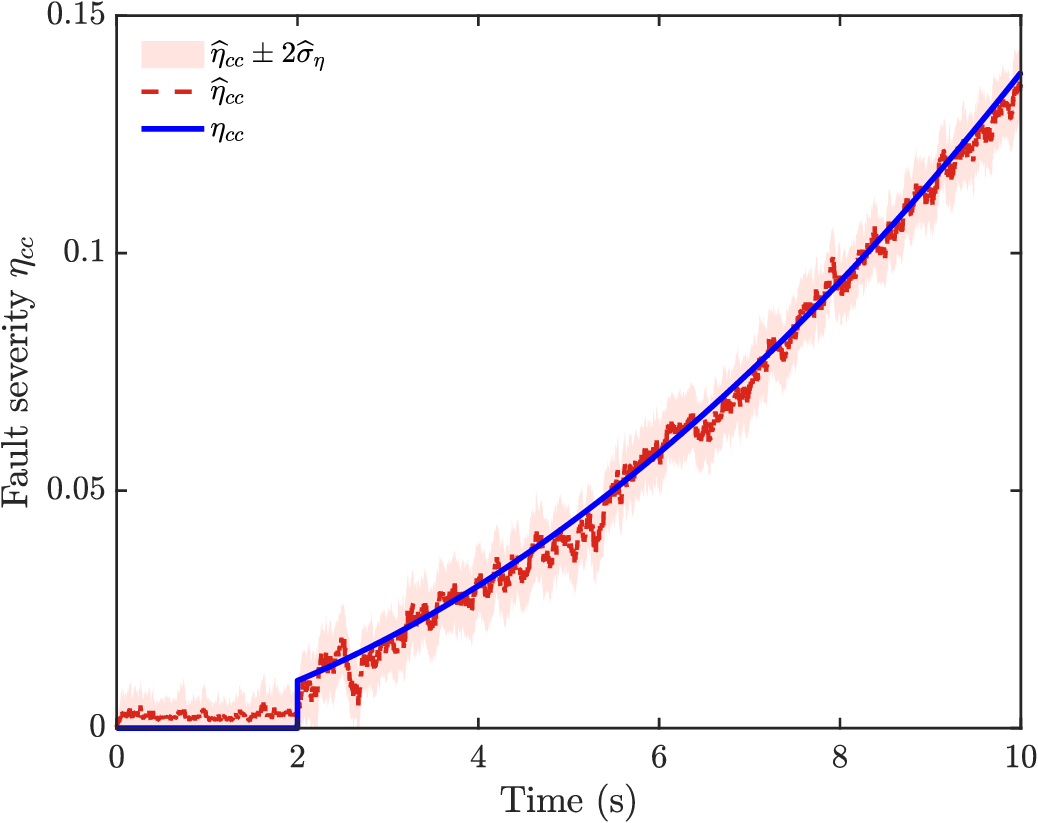}}
    \quad
    \subfloat[Estimation error.\label{fig:fig1b}]
    {\includegraphics[width=0.2575\linewidth]{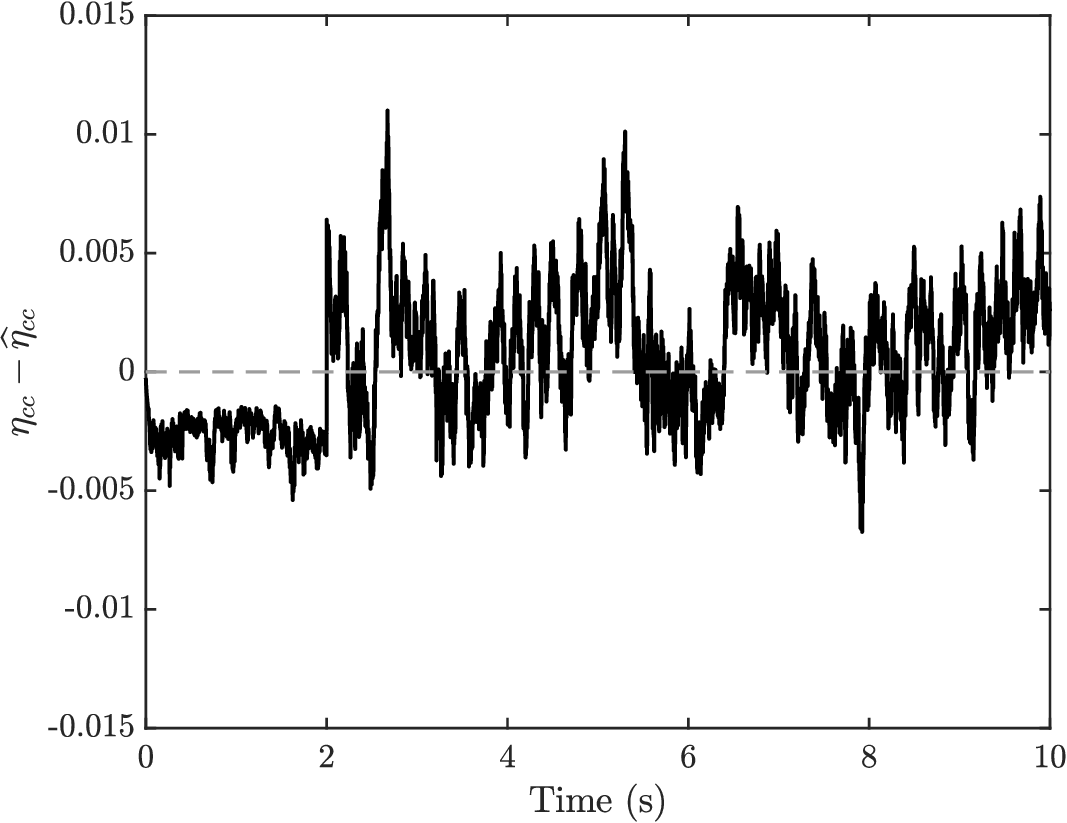}}
    \quad
    \subfloat[Effective sample size.\label{fig:fig1c}]
    {\includegraphics[width=0.25\linewidth]{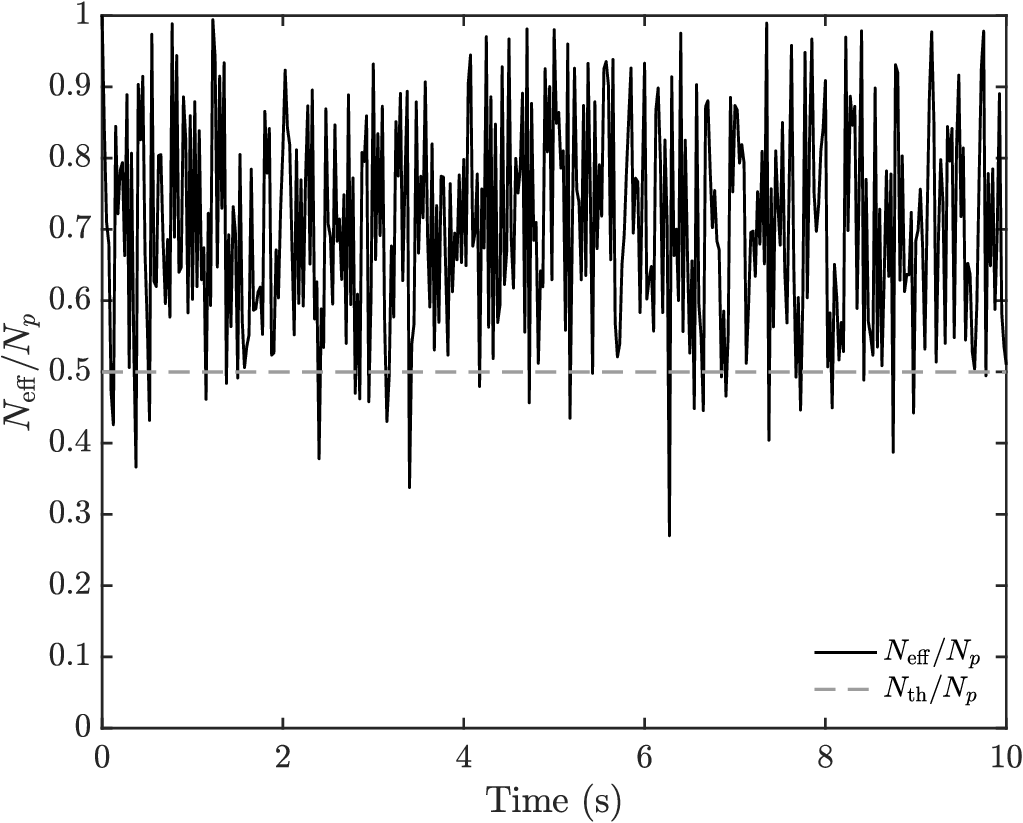}}
    \caption{Augmented-state particle-filter estimation of the stator
    inter-turn fault.}
    \label{fig:fig1}
\end{figure}

Figure~\ref{fig:fig1} evaluates the augmented-state PF of
Section~IV. The estimate $\widehat{\eta}_{cc}$ closely follows the true
progressive degradation and remains within the posterior
$\pm2\widehat{\sigma}_{\eta}$ band in Fig.~\ref{fig:fig1a}.
The resulting severity-estimation RMSE and MAE are
$2.8808\times10^{-3}$ and $2.3339\times10^{-3}$, respectively, while the
maximum absolute error is $1.1008\times10^{-2}$.
Figure~\ref{fig:fig1c} further shows the effective sample size in
\eqref{eq:effective_sample_size}; systematic resampling is activated whenever
$N_{\mathrm{eff},k}/N_p<0.5$. Over the complete simulation, $1436$
resampling events occur and the mean normalized effective sample size is
$0.716$.

\begin{figure}[t]
    \centering
    \subfloat[$d$-axis stator current.\label{fig:fig2a}]
    {\includegraphics[width=0.25\linewidth]{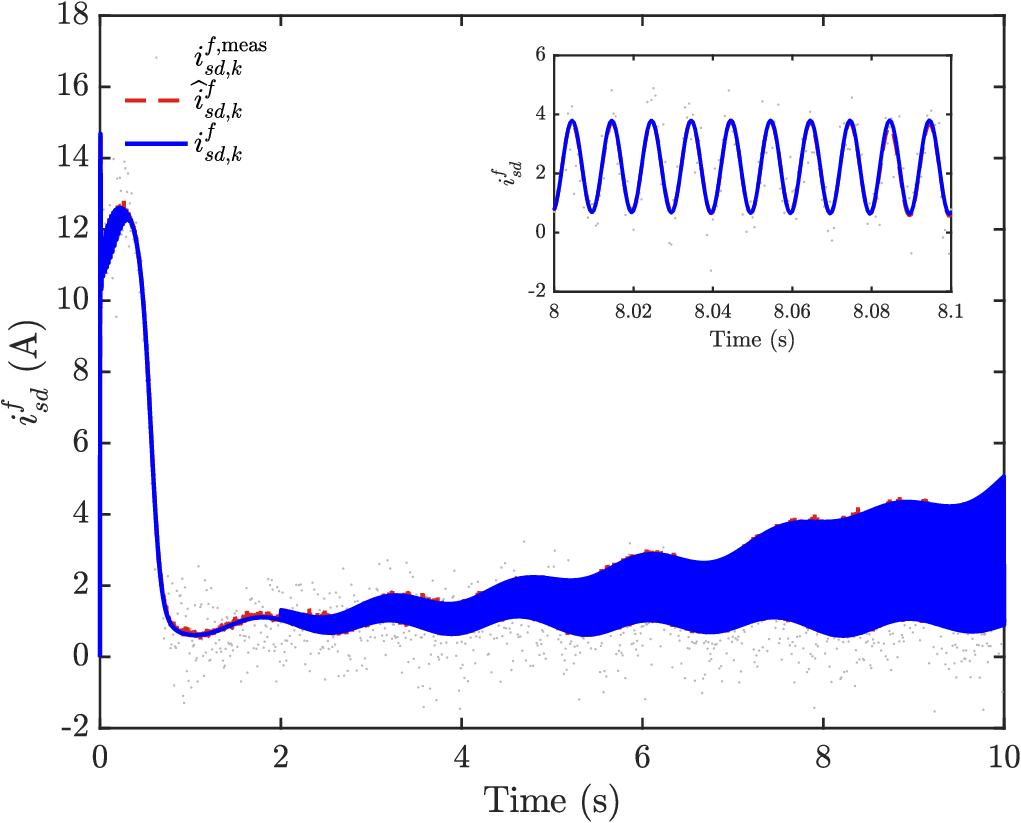}}
    \quad
    \subfloat[$q$-axis stator current.\label{fig:fig2b}]
    {\includegraphics[width=0.25\linewidth]{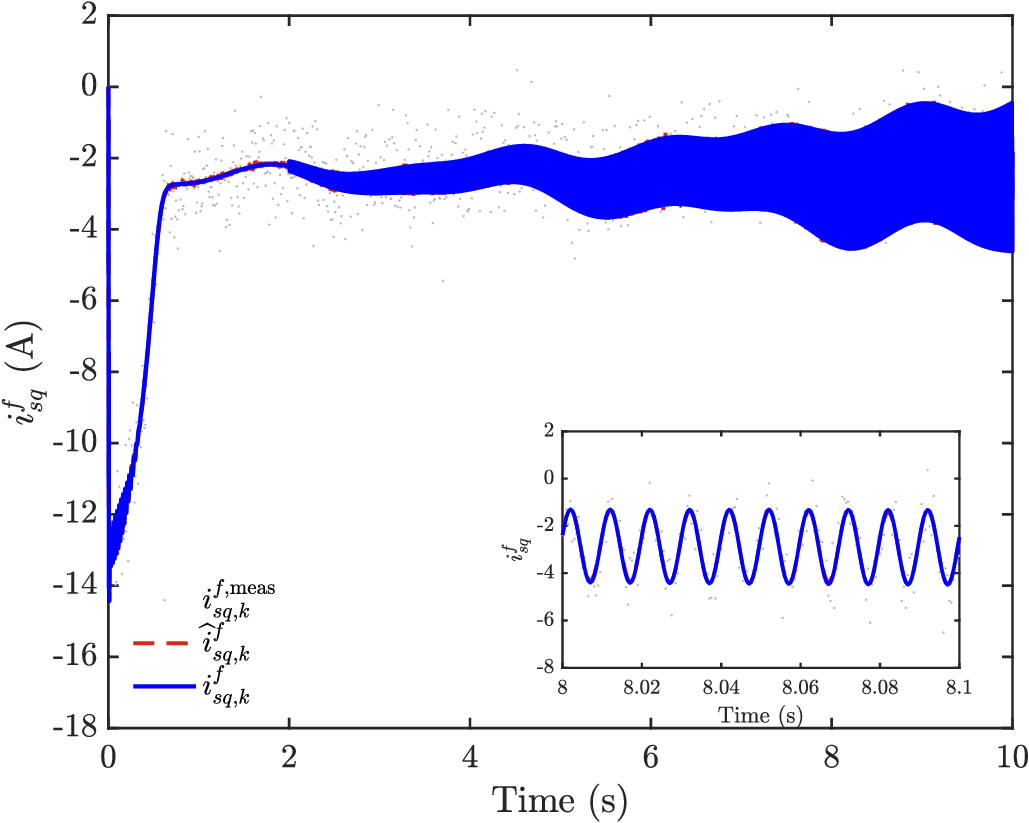}}
    \quad
    \subfloat[Rotor speed.\label{fig:fig2c}]
    {\includegraphics[width=0.25\linewidth]{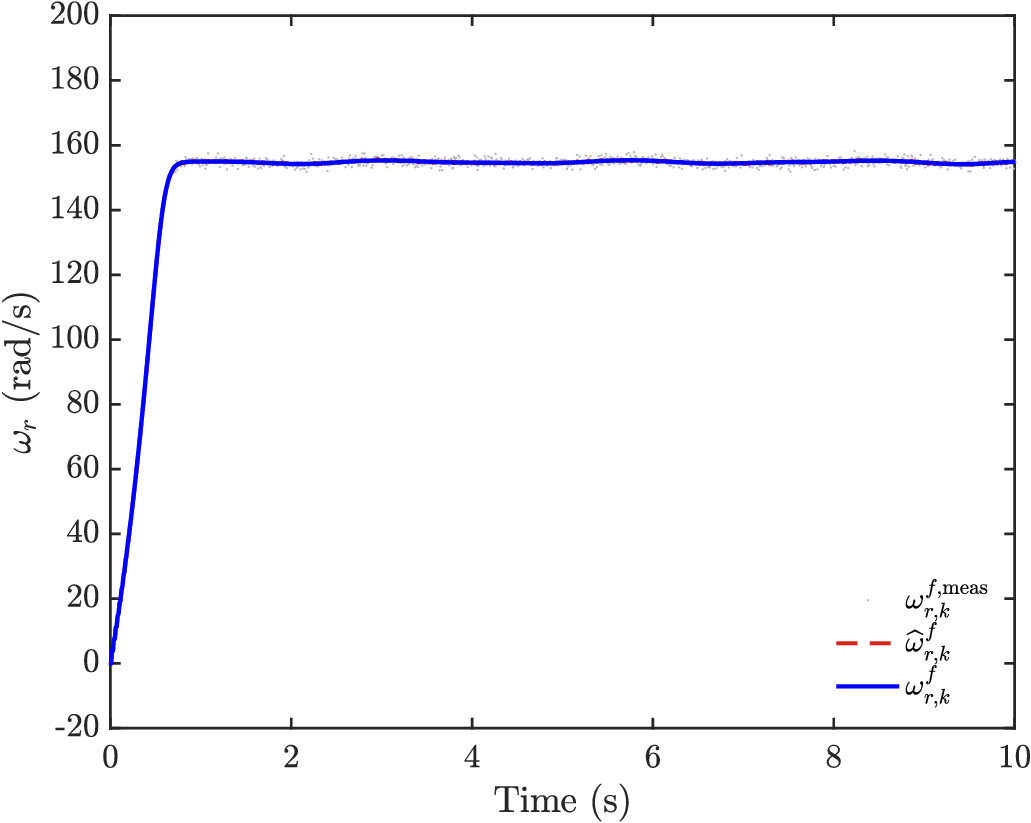}}
    \caption{Faulty-output reconstruction using the augmented-state PF.}
    \label{fig:fig2}
\end{figure}

Figure~\ref{fig:fig2} compares the measured, reconstructed, and true faulty
outputs generated by \eqref{eq:fault_output_affine}. The PF reconstructs
both stator-current components and the rotor speed despite measurement noise.
The enlarged current traces in Figs.~\ref{fig:fig2a}--\ref{fig:fig2b}
show that the reconstructed outputs follow the fast electrical oscillations,
supporting the joint state--fault representation in
\eqref{eq:augmented_state}.

\subsection{Fault Prognosis}

For prognosis, the PF degradation history is sampled every
$T_p=0.05$~s and the four models of Section~V are initialized from the
available health information. The nominal comparison uses the prognosis
instant $t_k=6.5$~s. The linear-trend model uses the most recent $30$
samples, Holt smoothing uses $\alpha=0.25$ and $\beta=0.12$, and the
Bayesian and particle forecasts use the most recent $12$ degradation
increments with $\mu_0=0$, $\sigma_{\mu,0}=2\times10^{-3}$, and
$\sigma_d=5\times10^{-4}$.

\begin{figure}[h!]
    \centering
    \subfloat[Linear trend.\label{fig:fig3a}]
    {\includegraphics[width=0.25\linewidth]{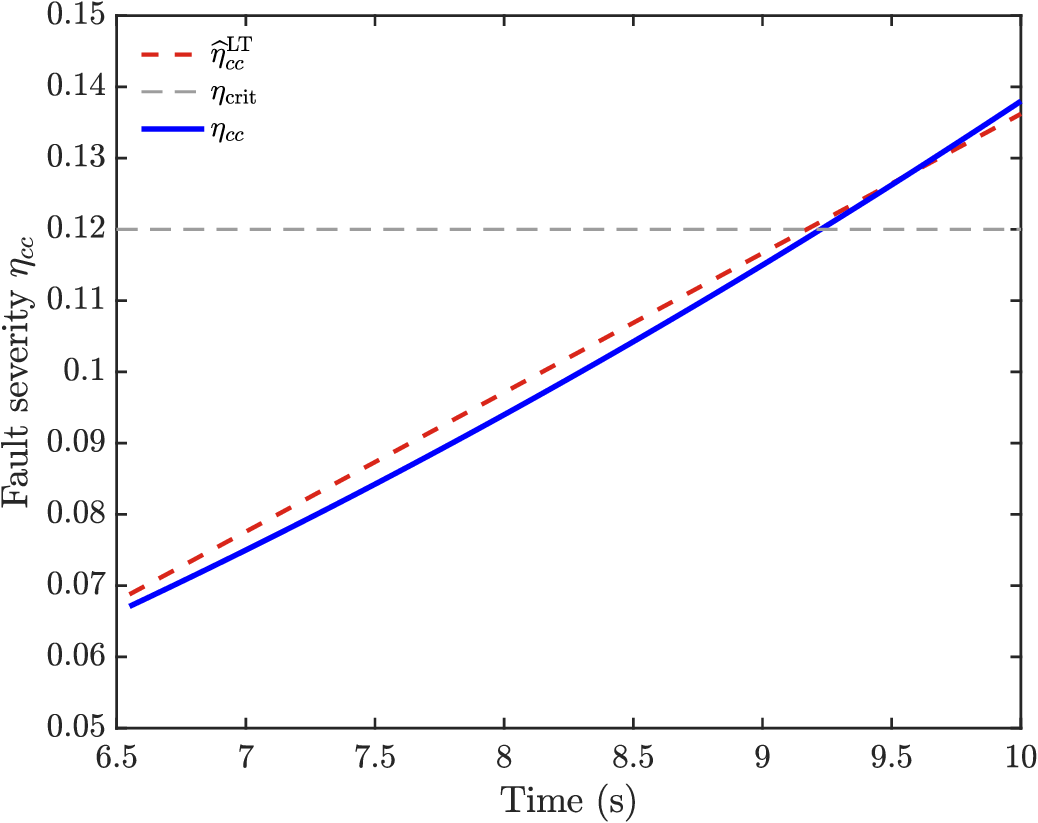}}
    \hfill
    \subfloat[Holt.\label{fig:fig3b}]
    {\includegraphics[width=0.25\linewidth]{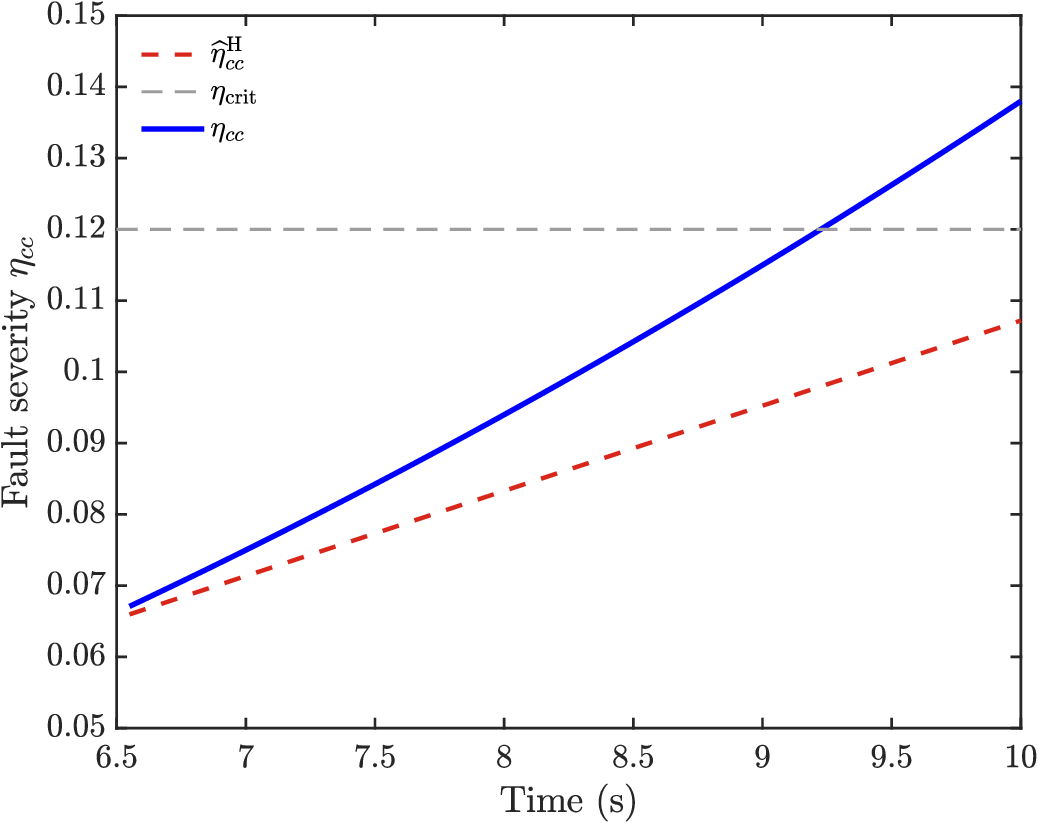}}
    \hfill
    \subfloat[Bayesian.\label{fig:fig3c}]
    {\includegraphics[width=0.25\linewidth]{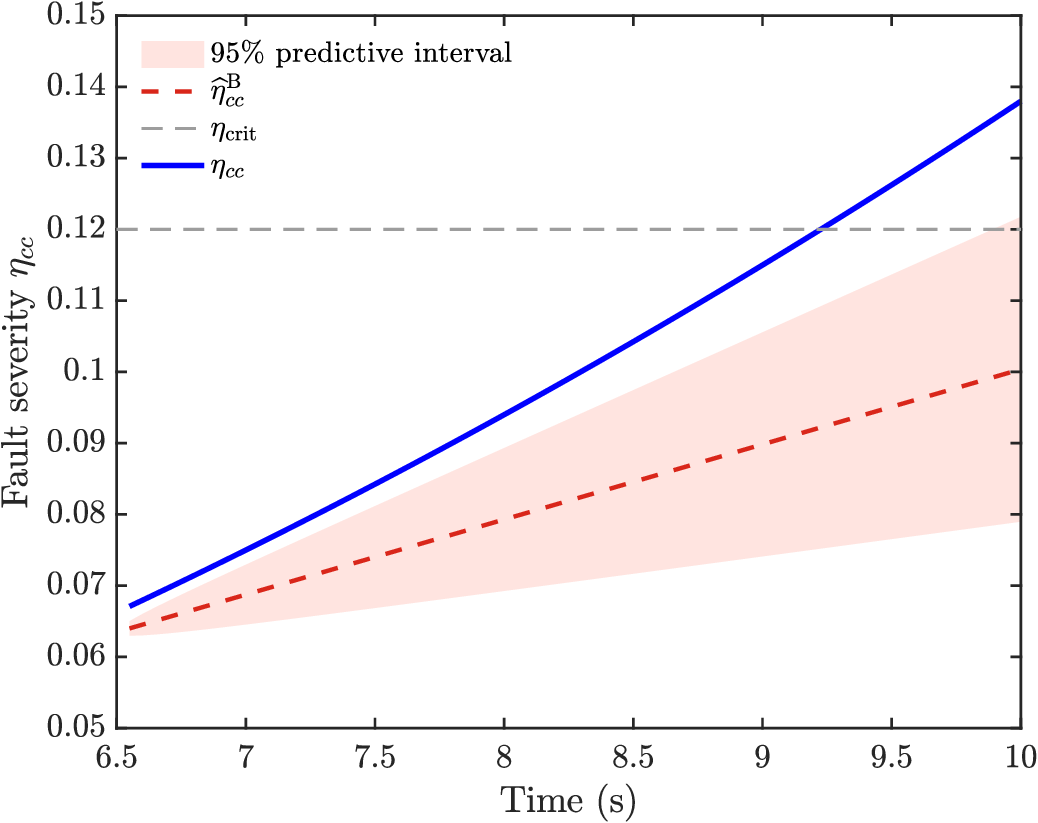}}
    \hfill
    \subfloat[Particle based.\label{fig:fig3d}]
    {\includegraphics[width=0.25\linewidth]{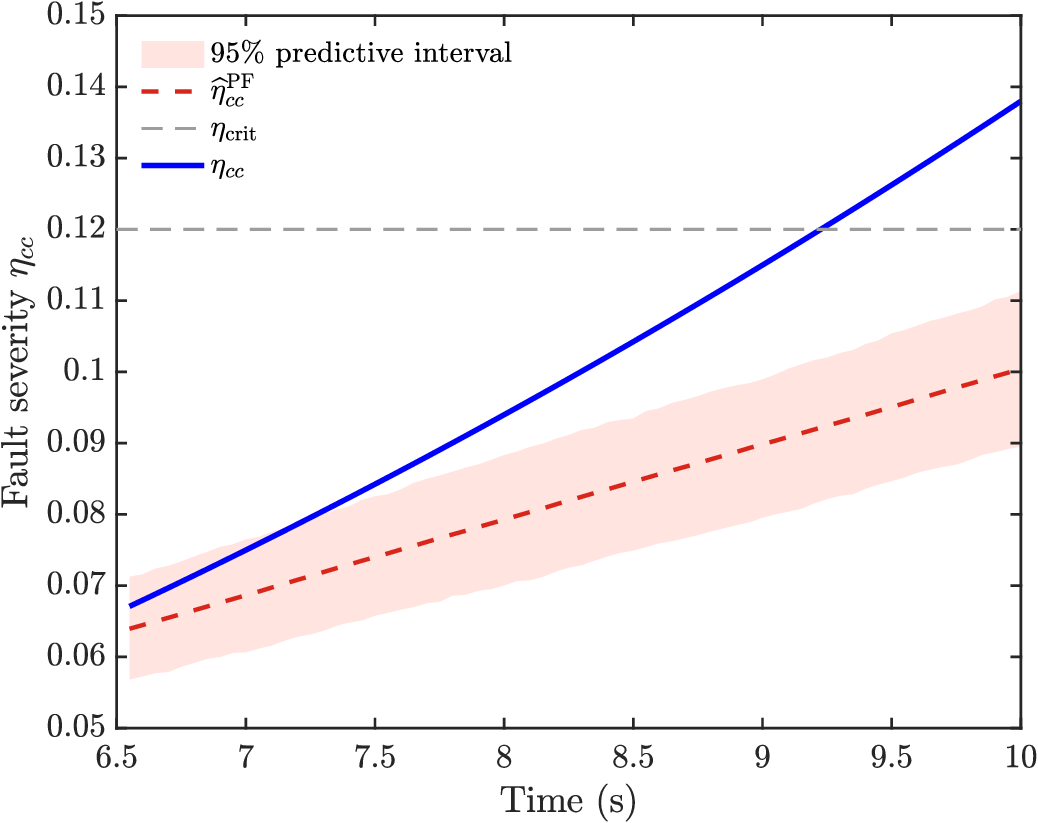}}
    \caption{Fault-severity prognosis from the common PF-based health
    information at $t_k=6.5$~s.}
    \label{fig:fig3}
\end{figure}

Figure~\ref{fig:fig3} compares the four prognostic descriptions of
$\eta_{cc}$. The linear-trend and Holt predictors provide the deterministic
forecasts in \eqref{eq:linear_prediction} and
\eqref{eq:holt_prediction}, whereas the Bayesian and particle-based methods
add predictive uncertainty through
\eqref{eq:bayesian_predictive_distribution} and
\eqref{eq:particle_predictive_distribution}. All four methods capture the
increasing degradation, while the probabilistic models additionally quantify
the growth of uncertainty with prediction horizon.

To examine prognosis beyond a single degradation trajectory, three patterns
are considered: gradual degradation,
$\eta_{cc}=0.01+0.012(t-2)$; accelerating degradation,
$\eta_{cc}=0.01+0.006(t-2)+0.0014(t-2)^2$; and an abrupt-change case
obtained from the gradual local evolution with an additional severity
increase of $0.025$ at $t=5.5$~s. Predictions are initiated at
$t_k\in\{3,5,7\}$~s and evaluated over a common $3$-s look-ahead horizon.

Figure~\ref{fig:fig4} shows how forecast accuracy depends jointly on the
prediction horizon, degradation pattern, and available history. For gradual
degradation, the locally fitted models remain accurate over the complete
horizon. Acceleration introduces increasing model mismatch as the horizon
grows, while the abrupt case is the most demanding because a future change
cannot be inferred from pre-change degradation history. The Bayesian and
particle-based mean forecasts remain close in all nine cases; their maximum
absolute differences range from $8.891\times10^{-5}$ to
$3.736\times10^{-4}$.

\begin{figure*}[h!]
    \centering
    \subfloat[Gradual, $t_k=3$ s.\label{fig:fig4a}]
    {\includegraphics[width=0.25\linewidth]{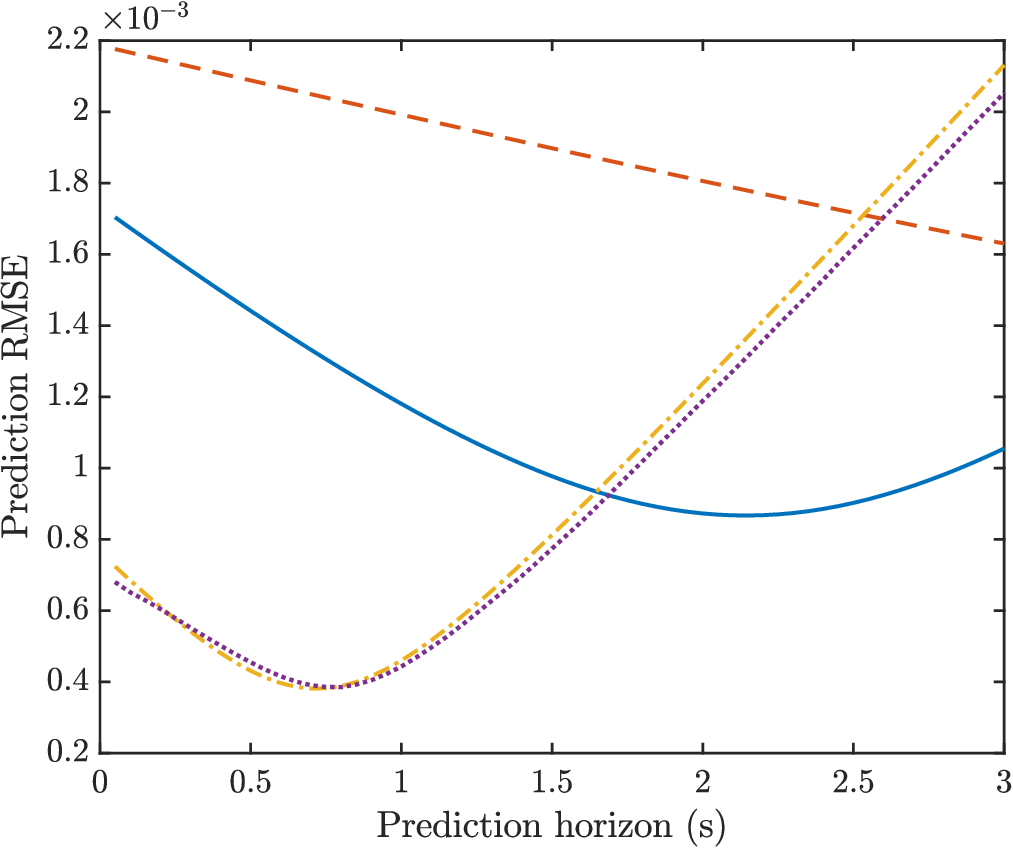}}
    \quad
    \subfloat[Accelerating, $t_k=3$ s.\label{fig:fig4b}]
    {\includegraphics[width=0.245\linewidth]{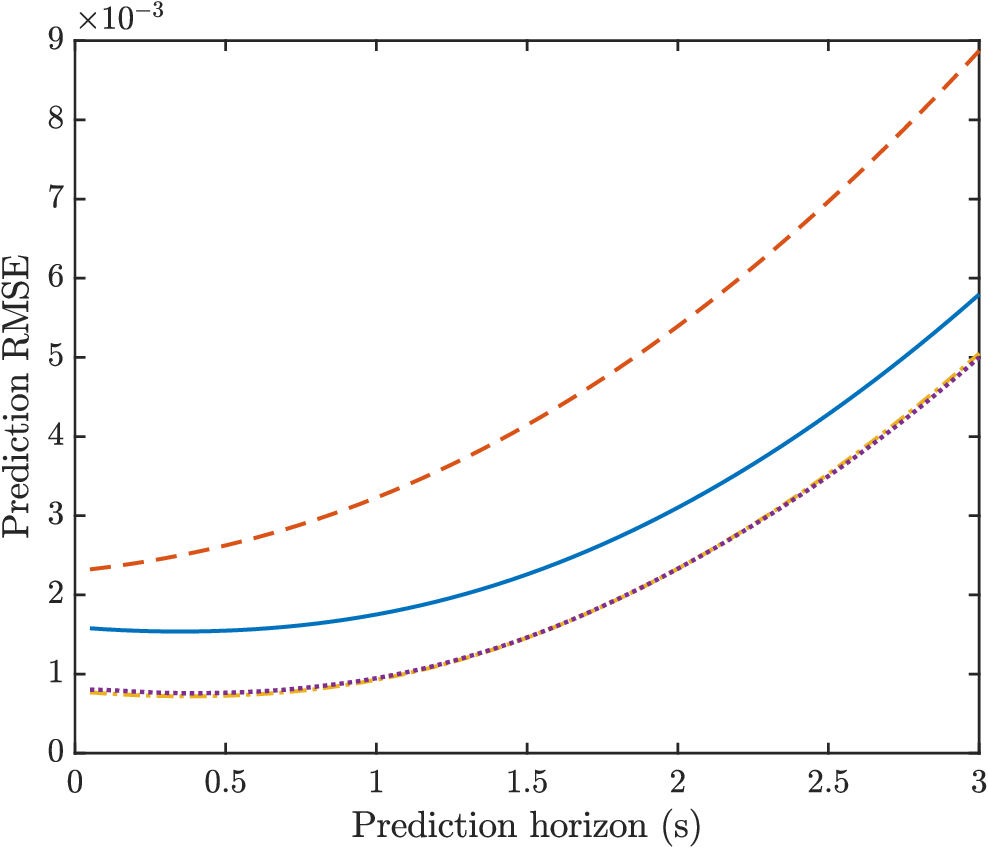}}
    \quad
    \subfloat[Abrupt, $t_k=3$ s.\label{fig:fig4c}]
    {\includegraphics[width=0.255\linewidth]{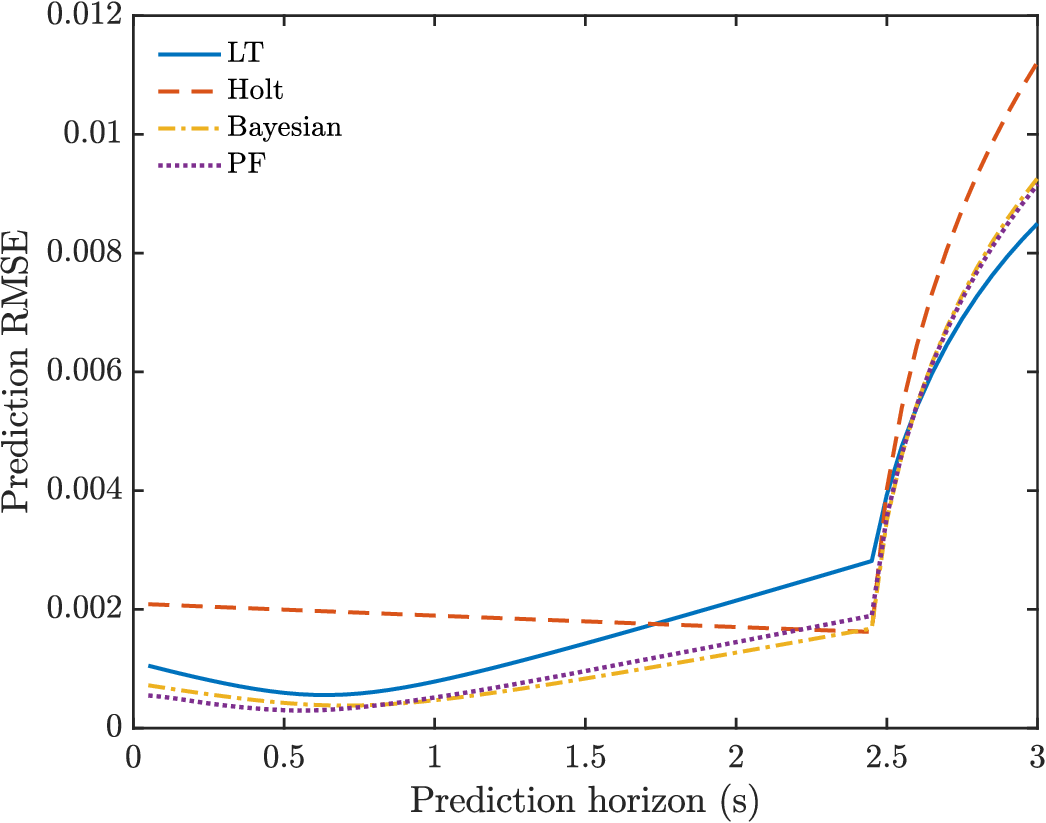}}
    \\[-1mm]
    \subfloat[Gradual, $t_k=5$ s.\label{fig:fig4d}]
    {\includegraphics[width=0.25\linewidth]{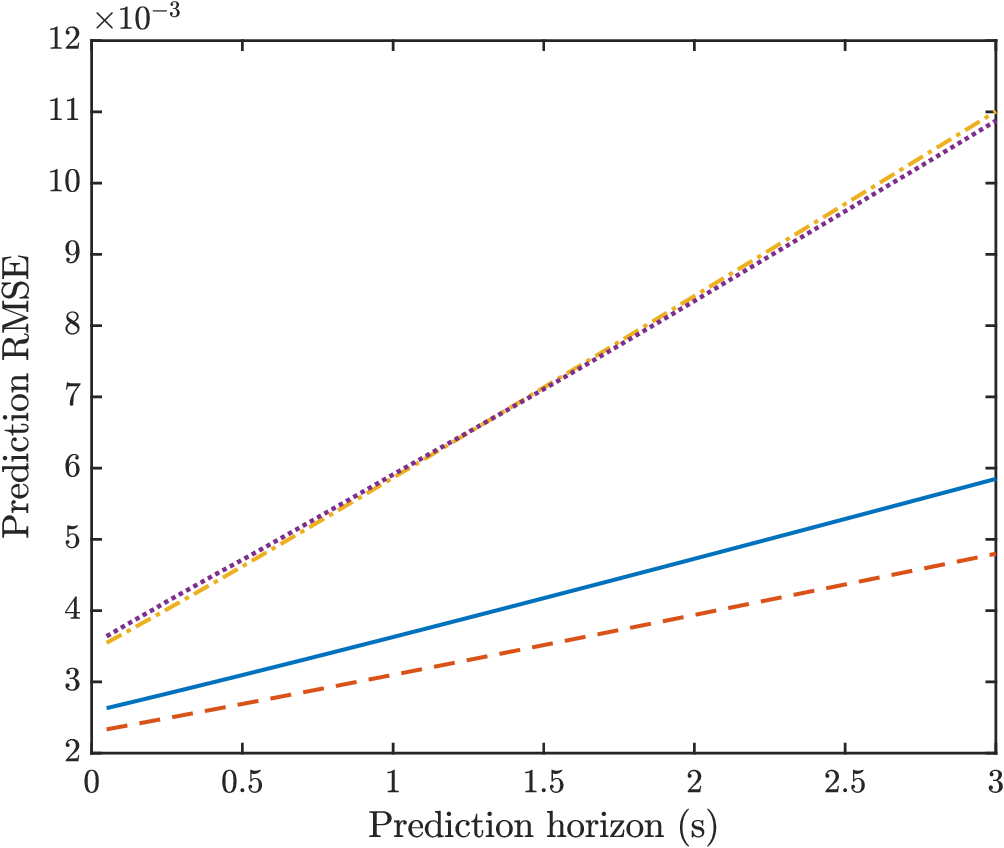}}
    \quad
    \subfloat[Accelerating, $t_k=5$ s.\label{fig:fig4e}]
    {\includegraphics[width=0.2575\linewidth]{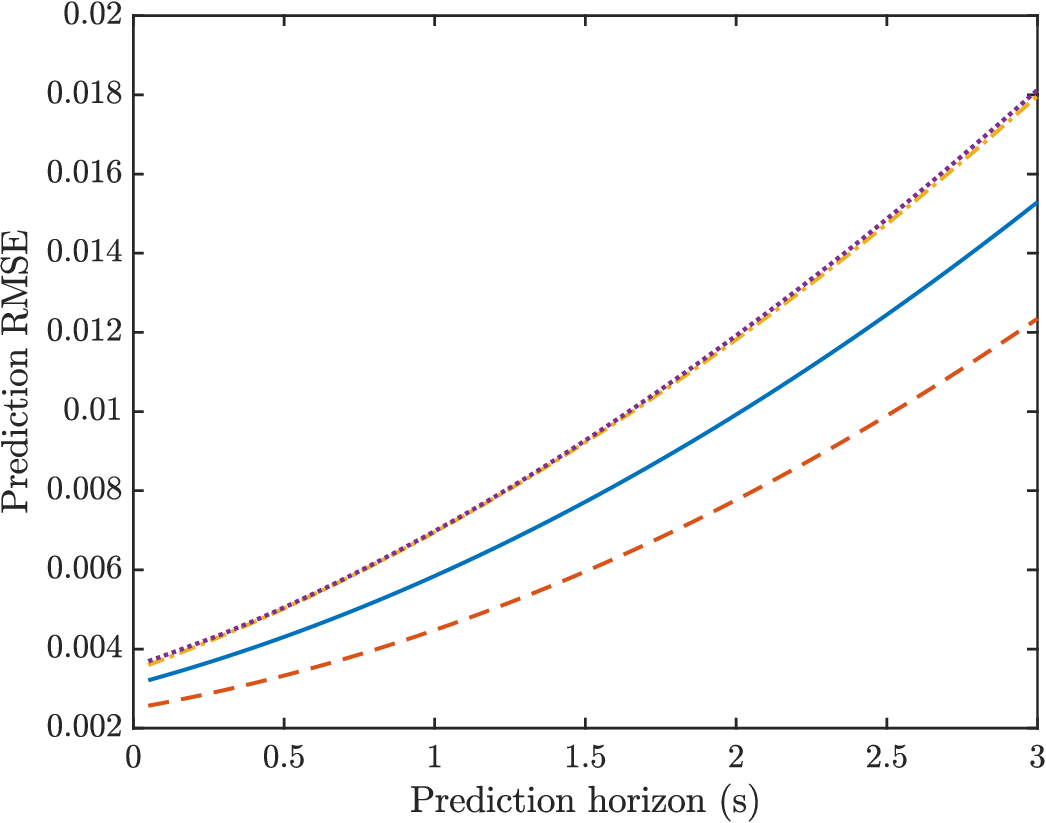}}
    \quad
    \subfloat[Abrupt, $t_k=5$ s.\label{fig:fig4f}]
    {\includegraphics[width=0.2575\linewidth]{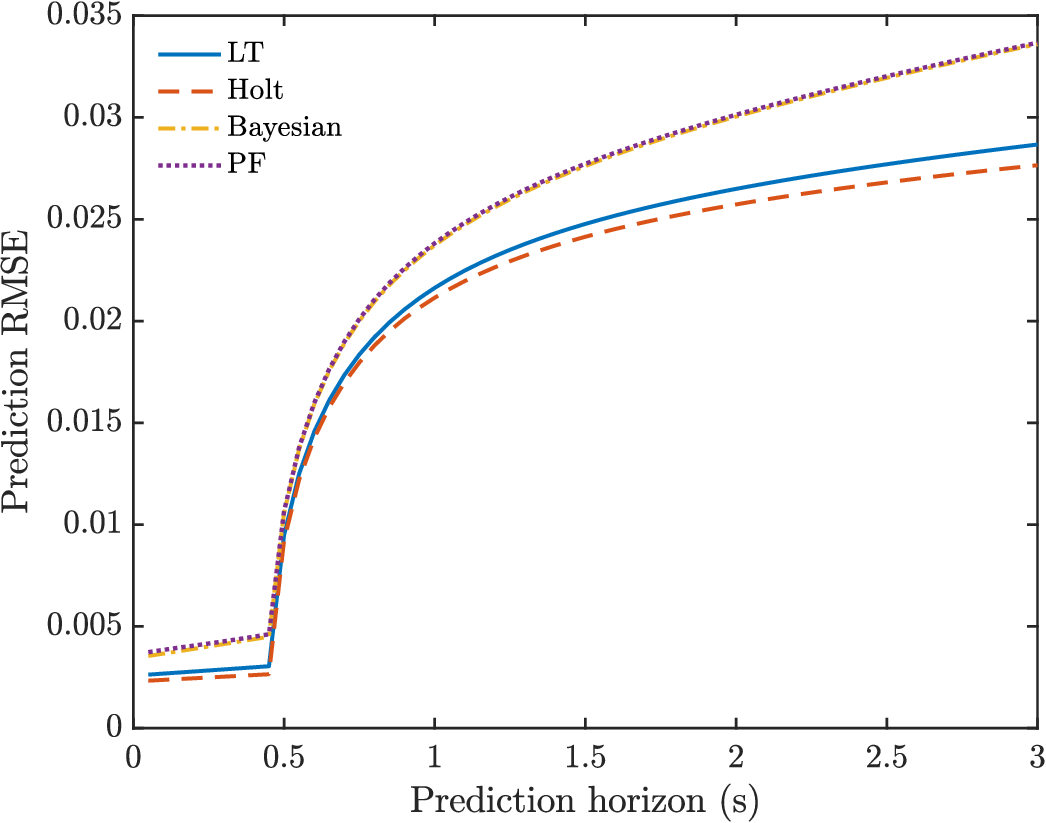}}
    \\[-1mm]
    \subfloat[Gradual, $t_k=7$ s.\label{fig:fig4g}]
    {\includegraphics[width=0.25\linewidth]{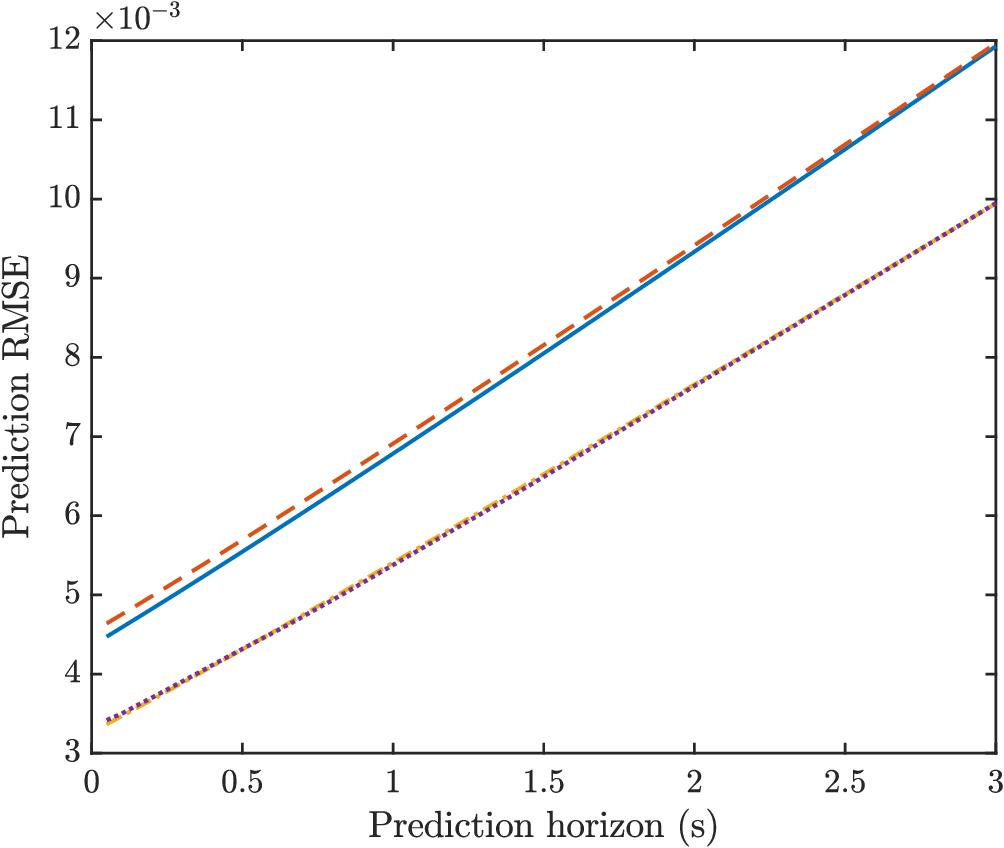}}
    \quad
    \subfloat[Accelerating, $t_k=7$ s.\label{fig:fig4h}]
    {\includegraphics[width=0.2575\linewidth]{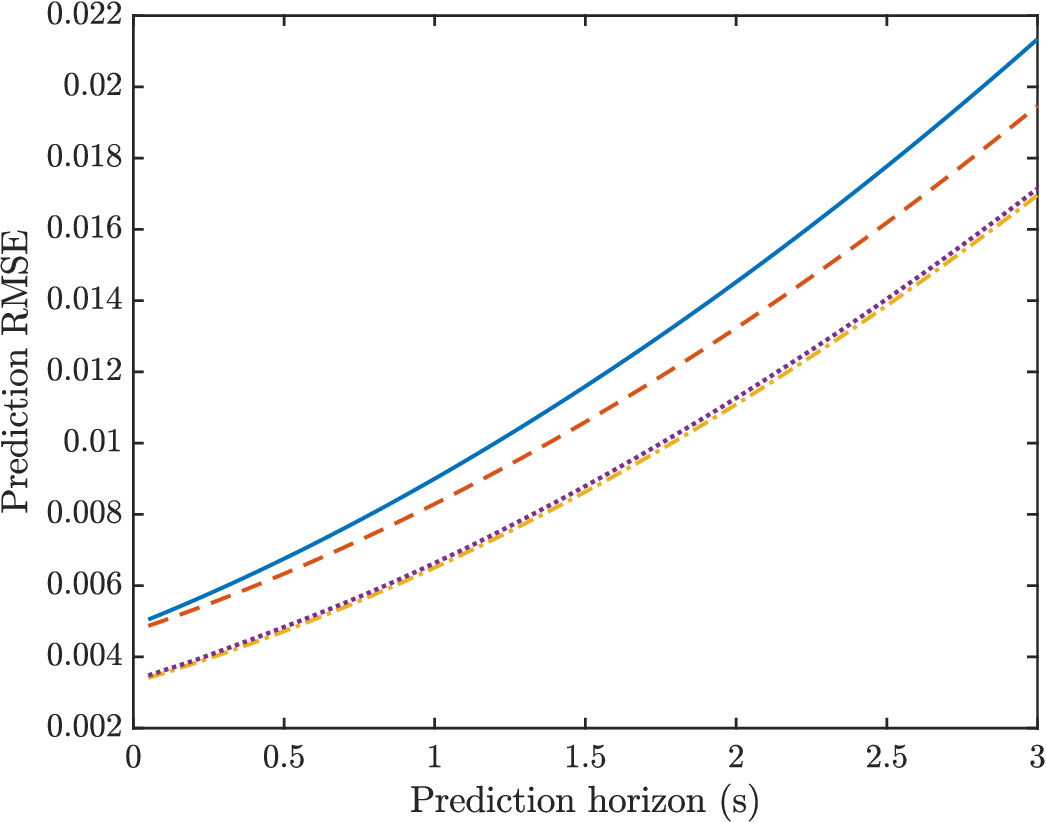}}
    \quad
    \subfloat[Abrupt, $t_k=7$ s.\label{fig:fig4i}]
    {\includegraphics[width=0.25\linewidth]{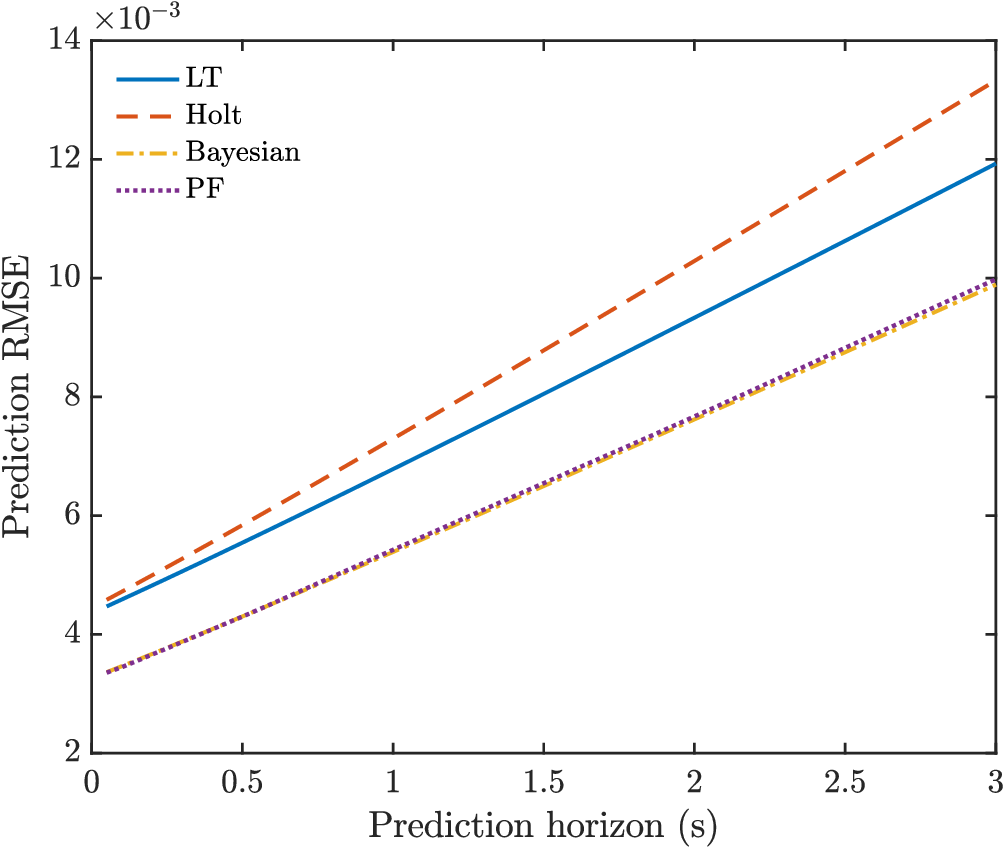}}
    \caption{Prediction RMSE versus look-ahead horizon for different
    degradation patterns and prognosis instants.}
    \label{fig:fig4}
\end{figure*}

\subsection{Threshold Crossing and Reliability}

Figures~\ref{fig:fig5}--\ref{fig:fig6} use a dedicated mildly accelerating
degradation trajectory beginning at $t_f=2$~s,
$\eta_{cc}=0.015+0.012(t-t_f)+3\times10^{-4}(t-t_f)^2$, with prognosis
initiated at $t_k=6.5$~s and critical severity
$\eta_{\mathrm{crit}}=0.12$. The corresponding health estimate is perturbed
by a small correlated error with standard deviation $5\times10^{-4}$ and
correlation coefficient $0.92$.

\begin{figure}[h!]
    \centering
    \includegraphics[width=0.5\linewidth]{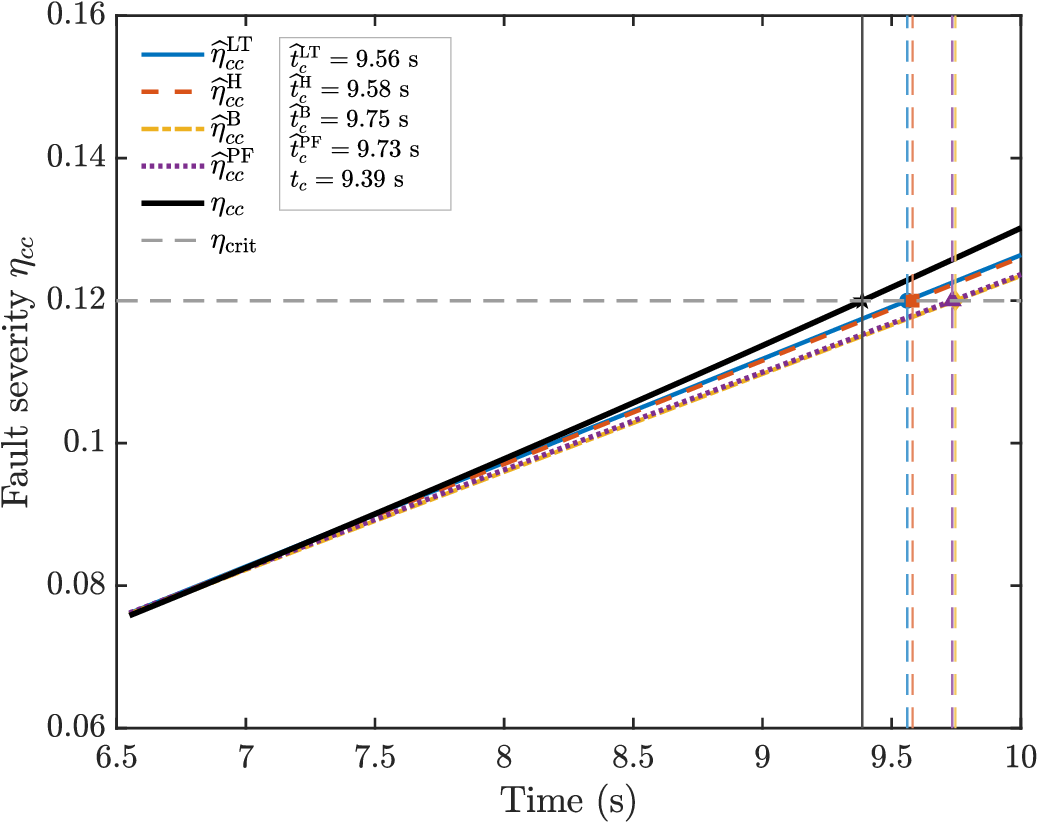}
    \caption{Predicted fault-severity trajectories, threshold crossings, and
    remaining useful life.}
    \label{fig:fig5}
\end{figure}

Figure~\ref{fig:fig5} applies the first-passage definition in
\eqref{eq:deterministic_failure_time}. The true degradation reaches
$\eta_{\mathrm{crit}}$ at $t_c=9.386$~s, corresponding to a true RUL of
$2.886$~s. The predicted crossings are $9.560$, $9.581$, $9.747$, and
$9.735$~s for the linear-trend, Holt, Bayesian, and particle-based models,
respectively, giving RUL predictions of $3.060$, $3.081$, $3.247$, and
$3.235$~s. Thus, all four forecasts identify the impending threshold
crossing, with the linear-trend prediction closest to the realized crossing
in this experiment.

\begin{figure}[h!]
    \centering
    \subfloat[First-passage reliability.\label{fig:fig6a}]
    {\includegraphics[width=0.25\linewidth]{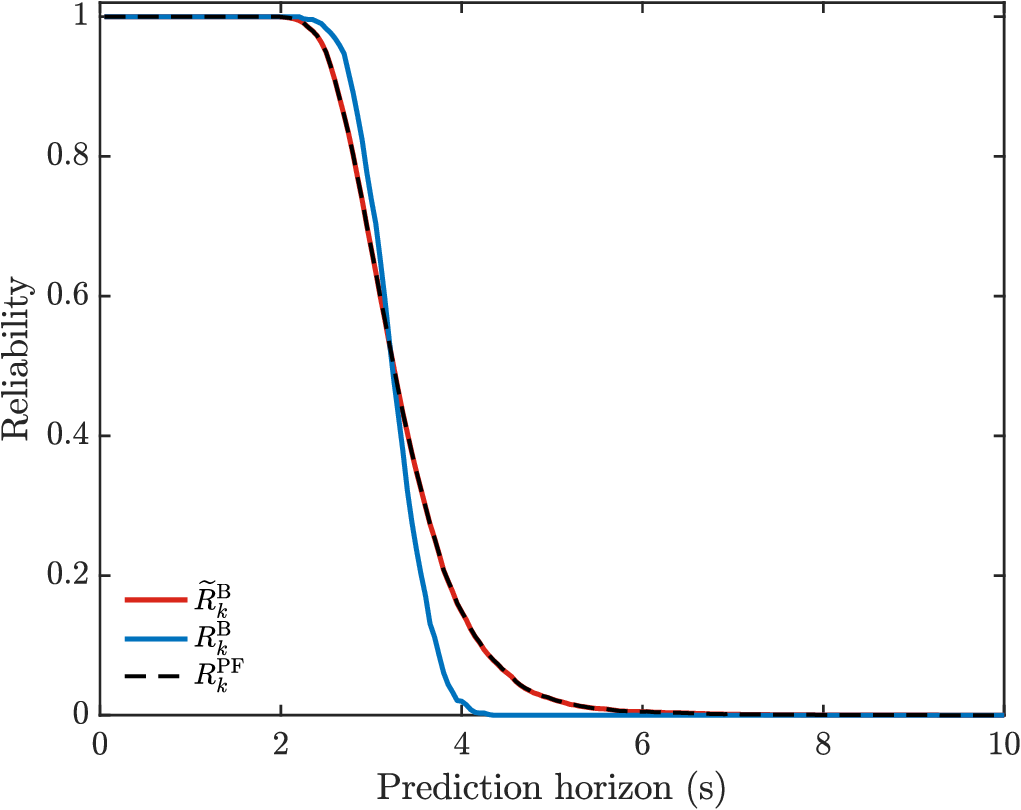}}
    \quad
    \subfloat[Hazard reliability.\label{fig:fig6b}]
    {\includegraphics[width=0.25\linewidth]{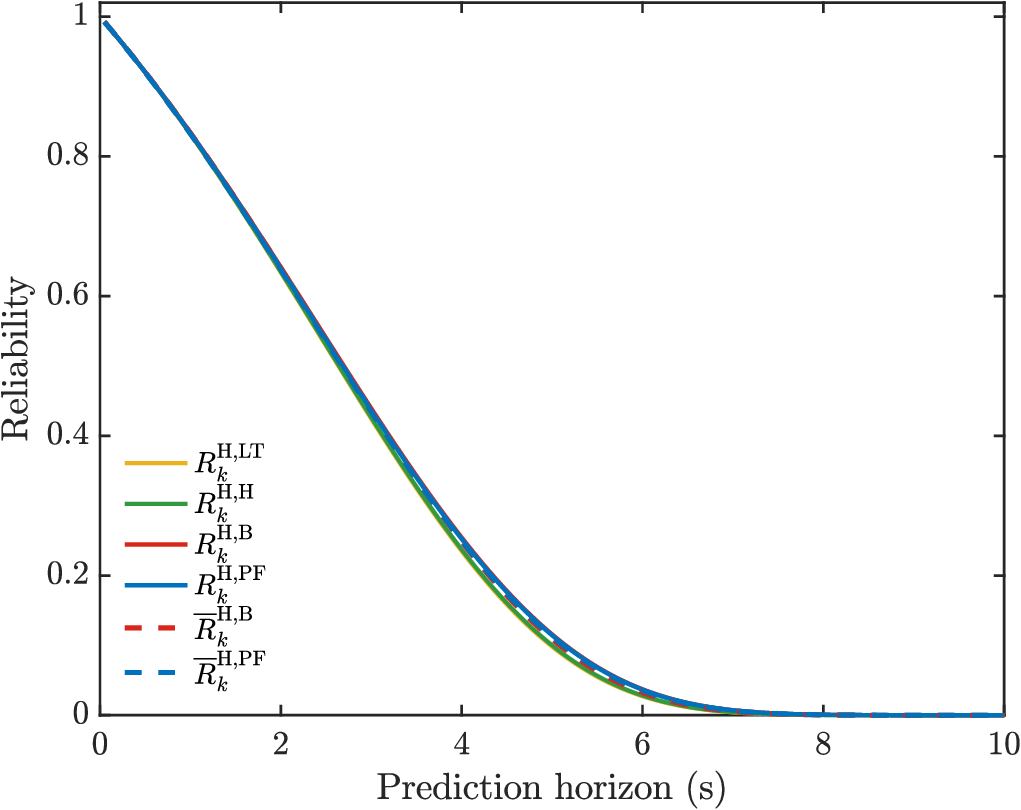}}
    \quad
    \subfloat[Weibull comparison.\label{fig:fig6c}]
    {\includegraphics[width=0.25\linewidth]{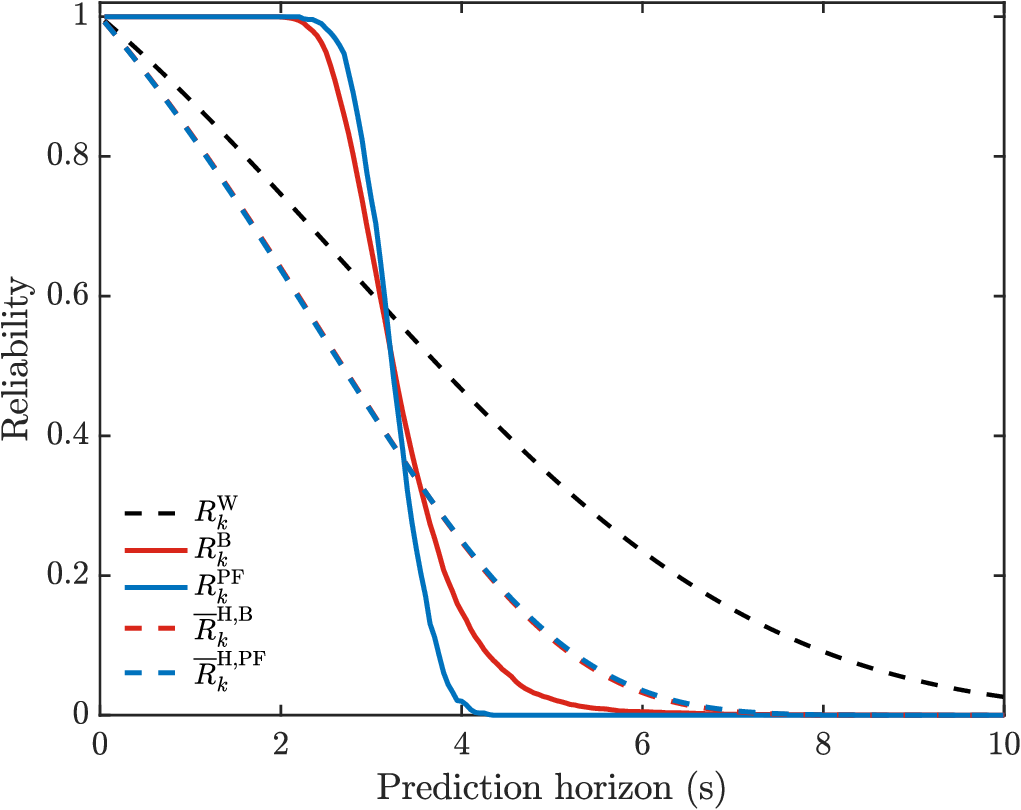}}
    \caption{Reliability assessment over a $10$-s prediction horizon.}
    \label{fig:fig6}
\end{figure}

Figure~\ref{fig:fig6a} illustrates the distinction between the marginal
Bayesian reliability $\widetilde{R}_k^{B}$ in
\eqref{eq:bayesian_pointwise_reliability} and the first-passage
reliabilities $R_k^{B}$ and $R_k^{\mathrm{PF}}$ in
\eqref{eq:bayesian_path_reliability} and
\eqref{eq:particle_path_reliability}. The pathwise quantities are more
restrictive because survival requires the complete predicted trajectory to
remain below $\eta_{\mathrm{crit}}$.

For Fig.~\ref{fig:fig6b}, the degradation-dependent hazard model uses
$h_0=0.02$ and $\alpha=3.2$. Reliability decreases as the accumulated
severity-dependent hazard in \eqref{eq:hazard_reliability_explicit}
increases. The uncertainty-aware Bayesian and PF curves additionally use
the predictive distributions through
\eqref{eq:bayesian_expected_hazard} and
\eqref{eq:particle_expected_hazard}, thereby separating point-forecast and
uncertainty-aware reliability.

Finally, Fig.~\ref{fig:fig6c} compares the degradation-informed measures
with the Weibull benchmark in \eqref{eq:weibull_reliability}, using
$\lambda_W=10.5$~s and $\beta_W=3$. The difference between these curves
emphasizes their distinct interpretations: Weibull reliability depends on
population-level lifetime statistics, whereas the Bayesian, particle-based,
and hazard formulations respond directly to the predicted winding condition.

\subsection{Sensitivity to Fault Orientation}

The final experiment evaluates whether PF estimation performance depends
strongly on the spatial orientation $\gamma_{cc}$. From
Proposition~\ref{prop:fault_spatial_periodicity}, the $36$-slot,
two-pole-pair motor has $n_e/(2p)=9$ distinct first-harmonic fault
orientations over $[0,90^\circ)$, namely
$\gamma_{cc}\in\{0^\circ,10^\circ,\ldots,80^\circ\}$. The complete
augmented-state PF is rerun for every orientation using five independent
Monte Carlo realizations.

\begin{figure}[t]
    \centering
    \includegraphics[width=0.5\linewidth]
    {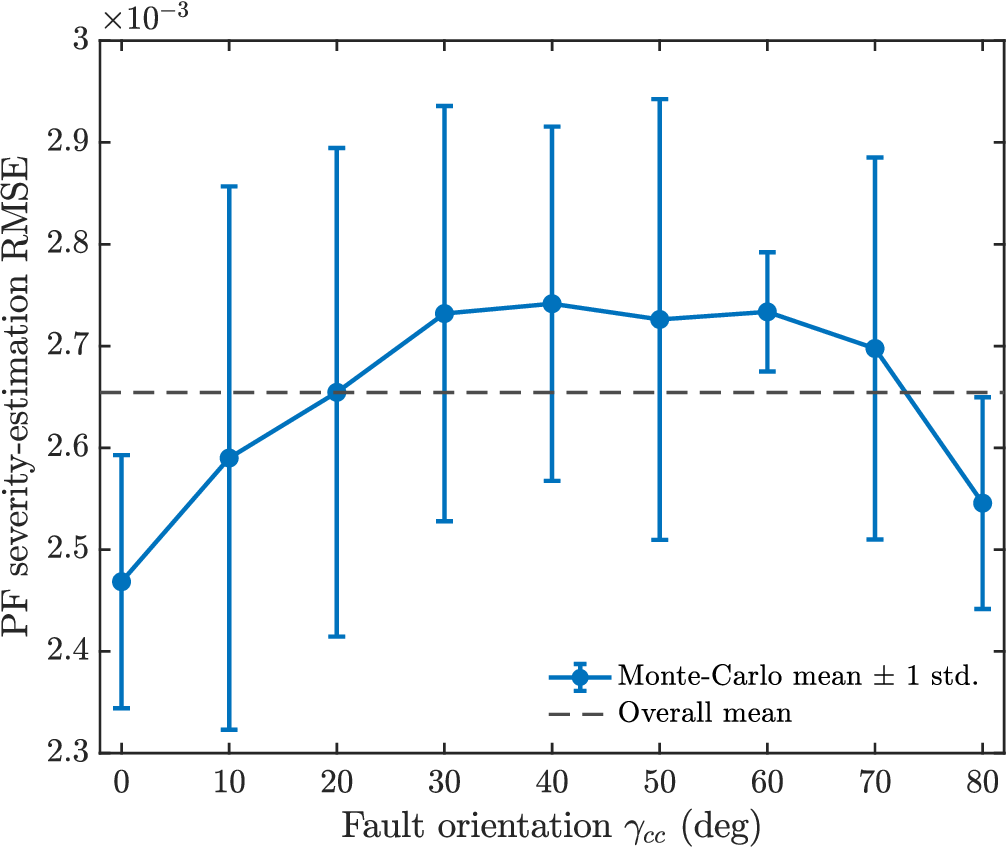}
    \caption{Sensitivity of fault-severity estimation to the spatial
    orientation $\gamma_{cc}$. Error bars denote one standard deviation over
    five Monte Carlo realizations.}
    \label{fig:fig7}
\end{figure}

Figure~\ref{fig:fig7} shows only moderate variation of estimation accuracy
with fault orientation. The mean RMSE ranges from
$2.4685\times10^{-3}$ at $\gamma_{cc}=0^\circ$ to
$2.7416\times10^{-3}$ at $\gamma_{cc}=40^\circ$, giving a worst-to-best
ratio of $1.111$. The overall mean RMSE is $2.6544\times10^{-3}$ and the
relative spread is $10.29\%$. Thus, although $\gamma_{cc}$ changes the
fault signature $\psi(t;\gamma_{cc})$, the augmented-state PF maintains
comparable severity-estimation accuracy across all nine distinct spatial
orientations.

\section{Conclusion}
\label{sec:conclusion}

This paper developed an integrated framework for stator inter-turn fault
estimation, prognosis, and reliability assessment in induction motors. The
fault was represented through a physically interpretable severity variable
$\eta_{cc}$ and spatial orientation $\gamma_{cc}$, leading to an output
model in which the fault contribution is affine in the unknown severity.
An augmented-state particle filter was then used to estimate jointly the
electromechanical state and winding degradation while preserving the physical
constraint $\eta_{cc}\in[0,1]$ and quantifying posterior uncertainty.

The estimated degradation history was subsequently propagated using linear
trend, Holt, Bayesian, and particle-based prognostic models. This enabled
both deterministic and probabilistic predictions of future fault severity,
which were then mapped into threshold-crossing RUL, first-passage reliability,
degradation-dependent hazard reliability, and a Weibull lifetime benchmark.
The resulting framework therefore connects online fault diagnosis directly
to forward-looking health assessment without changing the underlying physical
degradation variable.

Numerical results showed accurate fault-severity estimation and faulty-output
reconstruction, meaningful separation among prognostic models under gradual,
accelerating, and abrupt degradation, and consistent threshold-crossing and
reliability predictions. The spatial-orientation study further showed that
the particle-filter estimation accuracy remains comparable across all
distinct fault orientations implied by the periodic fault geometry. These
results support the use of the proposed framework as a unified probabilistic
pipeline from fault detection and estimation to prognosis and reliability
assessment.

Future work will extend the framework toward experimental validation under
nonstationary operating conditions and supply imbalance, where reliable ITSC
diagnosis remains challenging \cite{FW1}. Joint online inference of the fault
severity and spatial orientation will also be investigated, together with
physics-informed and hybrid data-driven prognostic models that can exploit
the physical degradation structure while adapting to limited fault data
\cite{FW2,Hedesh-CCTA}. Another direction is the integration of the proposed
estimation--prognosis pipeline with digital-twin architectures for continuous
model updating and life-cycle health assessment \cite{FW4,Wafi-SIAM,Wafi-SIAM-arXiv}. Finally,
transfer learning across operating regimes and the explicit conversion of
prognostic uncertainty into condition-based maintenance decisions will be
considered to improve generalization and operational usefulness
\cite{FW6,FW7}.

\newpage

\bibliographystyle{ieeetr}
\bibliography{reference}

\end{document}